\documentclass[a4paper,UKenglish,cleveref, autoref, thm-restate]{lipics-v2021}
\pdfoutput=1
\usepackage{lmodern}
\usepackage{cite}
\usepackage{amsmath}
\usepackage{amsfonts}
\usepackage{amssymb}
\usepackage{xcolor}
\usepackage{bussproofs}
\usepackage{amsthm}
\usepackage{boxedminipage}
\usepackage{proof}
\usepackage{mathtools}
\usepackage{stmaryrd} 
\usepackage{comment}
\usepackage{ifthen}

\newcommand\In[2]{#1(#2)}
\newcommand\Out[2]{\overline{#1}\langle{#2}\rangle}
\newcommand\tick{\mathtt{tick}}
\newcommand\PAR{\mathbin{\,|\,}}
\newcommand\pzero{0}
\newcommand\var{\textit{v}}
\newcommand\zero{\mathtt{0}}
\newcommand\suc{\mathtt{s}}
\newcommand{\minus}{\scalebox{0.75}[1.0]{$-$}}
\newcommand\seq[1]{\widetilde{#1}}
\newcommand\pifn[4]{\mathtt{match}~#1~\{\mathtt{case}~ \zero \mapsto #2 ;\mathtt{case}~ \suc(#3) \mapsto #4 \}}
\newcommand\congr{\equiv}
\newcommand\red{\rightarrow}
\newcommand\psub[3]{#1 [#2 := #3]}
\newcommand\Nat{\mathsf{Nat}}

\newcommand\subtype{\sqsubseteq}
\newcommand{\sem}[2]{\llbracket #2 \rrbracket_{#1}}

\newcommand\IV{\mathcal{V}}
\newcommand\Isub[3]{#1 \lbrace #3 / #2 \rbrace }
\newcommand\rel{\bowtie}
\newcommand\INN[1]{#1_{\NN}}
\newcommand\pred{\Rightarrow}
\newcommand\lcomp{\mathcal{C}_{\ell}}
\newcommand\IU[2]{\texttt{In}_{#2}^{#1}}
\newcommand\OU[2]{\texttt{Out}_{#2}^{#1}}
\newcommand\AU[3]{#1_{#3}^{#2}}
\newcommand\SubU{\sqsubseteq}
\newcommand\ZU{0}
\newcommand\ured{\longrightarrow}
\newcommand\Error{\mathtt{err}}
\newcommand\IPlus{\oplus}
\newcommand\Ilub{\sqcup}
\newcommand\Up[1]{{\uparrow^{#1}}}
\newcommand\ChT[2]{\texttt{ch}(#1)/#2}
\newcommand\ServT[4]{\ifthenelse{\equal{#1}{}}{\texttt{srv}^{#2}(#3)/#4}{\forall #1. \texttt{srv}^{#2}(#3) \ifthenelse{\equal{#4}{}}{}{/#4}}}
\newcommand\LeA{\mathsf{Left}}
\newcommand\RiA{\mathsf{Right}}

\newcommand\COL{\mathbin{:}}
\newcommand\C{\triangleleft}
\newcommand\p{\vdash}

\newcommand\q{\vDash}
\newcommand\nq{\nvDash}
\newcommand\fifi{\varphi;\Phi}
\newcommand\void{\cdot ; \cdot}
\newcommand\NN{\mathbb{N}}
\newcommand\NNi{\mathbb{N}_{\infty}}
\newcommand\set[1]{\{#1\}}
\newenvironment{framed}[0]{\begin{boxedminipage}{\linewidth}}{\end{boxedminipage}}
\newcommand\midd{\; \mbox{\Large{$\mid$}}\;}
\newcommand\vvskip{\vspace{3mm}}
\EnableBpAbbreviations 
\title{Sized Types with Usages for Parallel Complexity of Pi-Calculus Processes}

\author{Patrick Baillot}{Univ Lyon, CNRS, ENS de Lyon, Universite Claude-Bernard Lyon 1, LIP, F-69342, France}{}{}{}

\author{Alexis Ghyselen}{Univ Lyon, CNRS, ENS de Lyon, Universite Claude-Bernard Lyon 1, LIP, F-69342, France}{}{}{}

\author{Naoki Kobayashi}{The University of Tokyo, Japan}{}{}{}

\authorrunning{P. Baillot, A.Ghyselen and N. Kobayashi} 


\acknowledgements{We would like to thank anonymous referees for useful comments.
  This work was partially supported by the LABEX MILYON (ANR-10-LABX-0070)  of
  Universite de Lyon and JSPS KAKENHI Grant Number JP20H05703.}

\nolinenumbers 

\hideLIPIcs 

\EventEditors{Serge Haddad and Daniele Varacca}
\EventNoEds{2}
\EventLongTitle{32nd International Conference on Concurrency Theory (CONCUR 2021)}
\EventShortTitle{CONCUR 2021}
\EventAcronym{CONCUR}
\EventYear{2021}
\EventDate{August 23--27, 2021}
\EventLocation{Virtual Conference}
\EventLogo{}
\SeriesVolume{203}
\ArticleNo{23}
\begin{document}

\maketitle

\begin{abstract}
We address the problem of analysing the complexity of concurrent programs written
in Pi-calculus. We are interested in parallel complexity, or span,
understood as the execution time in a model with maximal
parallelism. A type system for parallel complexity has been recently
proposed by
the first two authors  
but it is too imprecise for non-linear channels and cannot analyse some concurrent processes.
Aiming for a more precise analysis, we design a type system which builds on the concepts of sized types and usages. The sized types allow us to parametrize the complexity by the size of inputs, and the usages allow us to achieve a kind of rely-guarantee reasoning on the timing each process communicates with its environment. We prove that our new type system soundly estimates the parallel complexity, and show through examples that it is often more precise than the previous type system of the first two authors.  
\end{abstract}

\section{Introduction} 
Static analysis of complexity is a classic topic of program analysis,
and various approaches to the complexity analysis, including type-based
ones~\cite{DBLP:conf/cav/0002AH12,HofmannJost03,HoffmannHofmann10,Hoffman12TPLS,DalLagoGaboardiLinearDependentTypes,AvanziniDalLagoAutomatingSizedTypeInference},
have been studied so far.
The complexity analysis of concurrent programs has been, however, much less studied.

In this paper, we are interested in
 analysing the parallel complexity (also called
\emph{span}) of the \(\pi\)-calculus, i.e.,
the maximal parallelized execution time under the assumption that an unlimited number
of processors are available~\cite{BaillotGhyselen21}.
%%We require that the analysis returns a \emph{parametric} complexity in the sense that
%%the complexity is parametrized by the size of inputs.
The parallel complexity should be parametrized by the size of inputs.
Following the success of previous studies on the complexity analysis of
sequential programs~\cite{DBLP:conf/cav/0002AH12,HofmannJost03,HoffmannHofmann10,Hoffman12TPLS,DalLagoGaboardiLinearDependentTypes,AvanziniDalLagoAutomatingSizedTypeInference}
and those on the analysis of other properties on
concurrent programs (e.g. \cite{DengSangiorgiTerminationTypability, DemangeonHKS07}), deadlock-freedom and livelock-freedom (e.g. \cite{Kobayashi97,Sumii98HLCL,KobayashietAlDeadlockFreeCalculus}; see
\cite{Kobayashi2003TypeSystemConcurrent} for a survey),
we take a type-based approach.

%Baillot and Ghyselen
The first two authors~\cite{BaillotGhyselen21} have actually
proposed a type-based analysis of the parallel complexity already.
However, as stressed by the authors, even though
their type system for span is useful
for analysing the complexity of some parallel programs, it %nevertheless
fails to type-check some examples of common concurrent programs, like semaphores. It is based on a combination of sized types and input/output types, in order to account suitably for the behaviour of channels w.r.t. reception and emission. 

In the present paper we design a type system for span which can deal with a much wider range of concurrent computation patterns
including the semaphores. For that, we take inspiration from the notion of  \textit{type usage}, which has been introduced and explored in \cite{Sumii98HLCL,KobayashietAlDeadlockFreeCalculus}, initially to guarantee absence of deadlock during execution. Type usages
are %can be thought of as
a generalization of input/output types, and % in the sense that they
describe how each channel is used for input and output.
Unlike the original notion of usages~\cite{Sumii98HLCL,KobayashietAlDeadlockFreeCalculus},
our usages are annotated with time intervals,
which are used for a kind of rely-guarantee reasoning, like ``assuming that
a message from the environment arrives during the time interval \([I_1,J_1]\), the process
sends back a message during the interval \([I_2,J_2]\)''; such reasoning is crucial for
analyzing the parallel complexity precisely and in a compositional manner.
We formalize the type system with usages and prove that it soundly estimates the parallel complexity.
%In the present paper we design a type system for span which can deal with a much wider range of concurrent computation patterns
%including the semaphores. For that we take inspiration from the notion of  \textit{type usage}, which has been introduced and explored in \cite{Sumii98HLCL,KobayashietAlDeadlockFreeCalculus}, initially to guarantee absence of deadlock during execution. Type usages
%are %can be thought of as
%a generalization of input/output types, and % in the sense that they
%describe how each channel is used for input and output.
%the multiple ways a given channel can be employed in a $\pi$-calculus term.
%This description is given as a kind of CCS process \cite{Milner89}. %We formalize the type system with usages and prove its soundness. We also show through examples (including semaphores) that our type system is often much more expressive than the type system of Baillot and Ghyselen~\cite{BaillotGhyselen21}.

\textbf{Contributions.}
The contributions of this paper are as follows. (i) The formalization of the new type system
for parallel complexity built on the new notion of usages: our usages are quite different
from the original ones, and properly defining them (including operations on
time intervals, usage reductions, and the notion of reliable usages) is non-trivial.
(ii) The proofs of type preservation and
complexity soundness: thanks to the careful definition of
new usages, the proofs are actually quite natural, despite the expressiveness of
the type system. (iii) Examples to demonstrate the precision and
expressive power of our new type system.
%We introduce a new notion of usage based on time intervals instead of numbers as in \cite{Sumii98HLCL,KobayashietAlDeadlockFreeCalculus}. We define the type system with usages and prove its complexity soundness. We also show through examples (including semaphores) that our type system is often much more expressive than the type system of \cite{BaillotGhyselen21}.  We believe that the new definition of usages we provide, its combination with sized types and their properties are technically challenging. In particular defining the right operations on intervals and  the conditions  for usage reductions leading to an adequate notion of reliable usages is non-trivial. We see as a benefit of our approach the fact that in the end the proofs of subject-reduction and complexity soundness for the type system are quite natural and understandable. This however crucially relies on several key original ingredients, including the choices mentioned above for the design of usages,  the definition of subusage relation rules and the operational semantics defining the span.  

%the multiple ways a given channel can be employed in a $\pi$-calculus term.
%%This description is given as a kind of CCS process \cite{Milner89}.\\ %We formalize the type system with usages and prove its soundness. We also show through examples (including semaphores) that our type system is often much more expressive than the type system of Baillot and Ghyselen~\cite{BaillotGhyselen21}.
%

\textbf{Paper outline.} We introduce in Sect.~\ref{s:picalculus} the $\pi$-calculus and the notion of parallel complexity we consider. Sect.~\ref{s:types} is devoted to the definition of types with usages. Then in Sect.~\ref{s:soundness} we prove the main result of this paper, the complexity soundness, and provide some examples. Finally, related work is discussed in Sect.~\ref{relatedwork}.

\section{The Pi-calculus with Semantics for Span} \label{s:picalculus}
In this section, we review
the definitions of the \(\pi\)-calculus and
its parallel complexity~\cite{BaillotGhyselen21}.

\subsection{Syntax and Standard Semantics for Pi-Calculus}
We consider a synchronous $\pi$-calculus, with a constructor $\tick$ that generates the time complexity.
The sets of \emph{variables}, \emph{expressions} and \emph{processes} are defined by:  
\begin{align*}
  \var \mbox{ (variables) }& := x,y,z \midd a,b,c \qquad
  e \mbox { (expressions) }:= \var \midd \zero \midd \suc(e) \\
	P \mbox{ (processes) }&:= \pzero \midd (P \PAR Q) \midd \In{a}{\seq{\var}}.P \midd !\In{a}{\seq{\var}}.P \midd \Out{a}{\seq{e}}.P  \midd (\nu a) P \\ &
        \phantom{:=} \midd \pifn{e}{P}{x}{Q} \midd \tick.P
\end{align*}
We use \(x,y,z\) as meta-variables for integer variables, and \(a,b,c\) as those for channel names.
The notation $\seq{\var}$ stands for a sequence of variables $\var_1,\var_2,\dots,\var_k$. Similarly, %we use
$\seq{e}$ denotes a sequence of expressions. We work up to $\alpha$-renaming, and write $\psub{P}{\seq{\var}}{\seq{e}}$ to denote the substitution of \(\seq{e}\) for the free variables \(\seq{\var}\) in \(P\).
For simplicity, we only consider integers as base types below, %in the following,
but the results can be generalized to other algebraic data-types such as lists or booleans.  

Intuitively, $P \PAR Q$ stands for the parallel composition of $P$
and $Q$. The process $\In{a}{\seq{\var}}.P$ represents an input: it stands for the
reception on the channel $a$ of a tuple of values identified by the variables $\seq {\var}$ in the continuation P. The process $!\In{a}{\seq{\var}}.P$ is a replicated version of $\In{a}{\seq{\var}}.P$: it behaves like an infinite number of $\In{a}{\seq{\var}}.P$ in parallel. The process $\Out{a}{\seq{e}}.P$ represents an output: it sends a sequence of expressions $\seq{e}$ on
the channel $a$, and continues as $P$.
We often omit \(\pzero\) and just write $\Out{a}{}$ for $\Out{a}{}.\pzero$.
A process $(\nu a) P$ dynamically creates a new channel name $a$ and then proceeds as $P$. We also have standard pattern matching on data types, and finally, the $\tick$ constructor incurs a cost of one in complexity but has no semantic relevance. 
%
%We consider that
This constructor is the only source of time complexity in a program. As the similar tick constructor in \cite{DasHoffmannPfenningTemporalSessionTypes}, it can represent different cost models and is more general than counting the number of reduction steps. For example, by adding $\tick$ after each input, we can count the number of communications in a process. By adding it after each replicated input on a channel $a$, we can count the number of calls to $a$. We can also count the number of reduction steps, by adding $\tick$ after each input and pattern matching. 

As usual, the structural congruence $\congr$ is 
defined as the least congruence containing:
$$
P \PAR \pzero \congr P \qquad P \PAR Q \congr Q \PAR P \qquad P \PAR (Q \PAR R) \congr (P \PAR Q) \PAR R $$ 
$$ (\nu a) (\nu b) P \congr (\nu b) (\nu a)P \qquad  (\nu a)(P \PAR Q) \congr (\nu a) P \PAR Q~(\text{when } a \text{ is not free in } Q) $$
Note that the last rule can always be applied from right to left by $\alpha$-renaming. By associativity, we will often write parallel composition for any number of processes and not only two.

\subsection{Parallel Complexity : The Span}

We now review the definition of span~\cite{BaillotGhyselen21}.
We add a process construct $m : P$, where $m$ is an integer.
A process using this constructor will be called an \emph{annotated process}.
Intuitively, this annotated process has the meaning \textit{$P$ with $m$ ticks before}.
The congruence relation $\congr$ is then enriched with the following relations: 
$$ 0 : P \congr P \qquad m : (P \PAR Q) \congr (m : P) \PAR (m : Q)  $$
$$ m : (\nu a) P \congr (\nu a) (m : P) \qquad m : (n : P) \congr (m + n) : P \qquad  $$
So, zero tick is equivalent to nothing and ticks can be distributed over parallel composition as expressed by the second relation. Name creation can be done before or after ticks without changing the semantics and finally ticks can be grouped together.

\begin{comment} 
With this congruence relation and this new constructor, the canonical form presented in Definition~\ref{d:guardedcanonical} is given a new shape. 

\begin{definition}[Canonical Form for Annotated Processes]
	An annotated process is in canonical form if it has the shape: $(\nu \seq {a}) (n_1 : G_1 \PAR \cdots \PAR n_m : G_m)$ with $G_1,\dots,G_m$ guarded processes. 
\end{definition} 
\end{comment} 
%Finally,
The rules for the reduction relation $\pred$ are given in Figure~\ref{f:parallelreduction}.
%Intuitively,
This semantics works as the usual semantics for $\pi$-calculus, but when doing a synchronization, only the maximal annotation is kept, and ticks are memorized in the annotations. 
\begin{figure*}
	\centering
	\begin{framed}
		\begin{center}
			%Synchro  
			\vvskip
			\AXC{}
			\UIC{$(n:\In{a}{\seq{\var}}.P) \PAR (m:\Out{a}{\seq{e}}.Q) \pred \max(m,n):(\psub{P}{\seq{\var}}{\seq{e}} \PAR Q)$}
			\DP
			\qquad 
			\AXC{} 
			\UIC{$\tick. P \pred 1 : P$}
			\DP 
			\\
			\vvskip
			% Server
			\AXC{}
			\UIC{$(n:!\In{a}{\seq{\var}}.P) \PAR (m:\Out{a}{\seq{e}}.Q) \pred (n:!\In{a}{\seq{\var}}.P) \PAR (\max(m,n):(\psub{P}{\seq{\var}}{\seq{e}} \PAR Q))$} 
			\DP
			\\
			\vvskip
			%Match Integers
			\AXC{}
			\UIC{$\pifn{\zero}{P}{x}{Q} \pred P$}
			\DP 
			\\ 
			\vvskip  
			\AXC{}
			\UIC{$\pifn{\suc(e)}{P}{x}{Q} \pred \psub{Q}{x}{e}$}
			\DP
			\\ 
			\vvskip
			%Context Rules
			\AXC{$P \pred Q$}
			\UIC{$ P \PAR R \pred Q \PAR R$}
			\DP 
			\qquad
			\AXC{$P \pred Q$}
			\UIC{$(\nu a) P \pred (\nu a) Q$}
			\DP 
			\qquad 
			\AXC{$P \pred Q$}
			\UIC{$(n : P) \pred (n : Q)$}
			\DP 
			\\ 
			\vvskip   
			\AXC{$P \congr P'$}
			\AXC{$P' \pred Q'$}
			\AXC{$Q' \congr Q$}
			\TIC{$P \pred Q$}
			\DP  
		\end{center}   
	\end{framed}
	\caption{Reduction Rules for Annotated Processes}
	\label{f:parallelreduction}
\end{figure*}
Span is then defined by: 
\begin{definition}[Parallel Complexity]
The \emph{local complexity} $\lcomp(P)$
  of an annotated process \(P\) is defined by:  
	$$\lcomp(n : P) = n + \lcomp(P)\qquad \lcomp(P \PAR Q) = \max(\lcomp(P),\lcomp(Q)) $$ 
	$$\lcomp((\nu a) P) = \lcomp(P) \qquad \lcomp(P) = 0 \text{ otherwise} $$
	%Equivalently, $\lcomp(P)$ is the maximal integer that appears in the canonical form of $P$. 
        The \emph{global parallel complexity} (or \emph{span}) of \(P\)
        is given by $\max \{ n \mid P \pred^* Q \land \lcomp(Q) = n \}$ where $\pred^*$ is the reflexive and transitive closure of $\pred$.
	\label{d:parallelcomplexity}  
\end{definition}

\begin{comment} 
This parallel complexity satisfies the following lemma \cite{BaillotGhyselen21}:

\begin{lemma}[Reduction and Local Complexity]
	Let $P,P'$ be annotated processes such that $P \pred P'$. Then, $\lcomp(P') \ge \lcomp(P)$ 
\end{lemma}  

So, in order to bound the complexity of an annotated process, we need to reduce it with $\pred$, and then take the maximal local complexity over all normal forms. If there is none, the complexity is the limit of the local complexity on those infinite reductions. One can see that this semantics respects the conditions given above. 
\end{comment} 

\begin{example}
%	\AG{Maybe this example can be removed} 
%%        Consider this process, a kind of simplified semaphore:
        Let $P := \tick. \In{a}{}. \tick. \Out{a}{} \PAR \tick. \In{a}{}. \tick. \Out{a}{} \PAR \Out{a}{}$. Then, we have:
	\begin{align*}
	P &\pred^2 1 : (\In{a}{}. \tick. \Out{a}{}) \PAR 1 : (\In{a}{}. \tick. \Out{a}{}) \PAR 0 : \Out{a}{} \pred 1 : (\In{a}{}. \tick. \Out{a}{}) \PAR 1 : (\tick. \Out{a}{}) \\
	&\pred 1 : (\In{a}{}. \tick. \Out{a}{}) \PAR 2 : \Out{a}{} \pred 2 : (\tick. \Out{a}{}) \pred 3 : \Out{a}{}
	\end{align*}
        %	So
        Thus, the process has at least complexity $3$. As all the other possible choices we could have made in the %%order of
        reduction steps are similar, the process has exactly complexity $3$.
\label{e:semaphore}
\end{example} 

The following example motivates our introduction of usages in the next section.

\begin{example}[Motivating Example]
  \label{ex:motivating}
	Let $P := \In{a}{}. \tick. \Out{a}{}$.
	Then, the complexity of $P \PAR P \PAR P \PAR \cdots \PAR P \PAR \Out{a}{}$ is equal to the number of $P$ in parallel. 
\end{example}

\section{Types with Usages} \label{s:types} 

The goal of our work is to design a type system for processes such that if $\Gamma \p Q \C K$ then $K$ is a bound on the complexity of $Q$, as in \cite{BaillotGhyselen21}.
The analysis of \cite{BaillotGhyselen21} was not precise enough:
in fact, the process \(P\) in Example~\ref{ex:motivating} was not typable.
%%In addition to sized types \cite{HughesParetoSabrySizedTypes,DalLagoGaboardiLinearDependentTypes,AvanziniDalLagoAutomatingSizedTypeInference}, the type system we present relies on usages \cite{Kobayashi2003TypeSystemConcurrent} to tackle this problem.
The main idea to tackle this problem
is to use the notion of usages to represent the channel-wise behaviour
of processes. Usages have been used for deadlock-freedom analysis \cite{Sumii98HLCL,KobayashiTypeSystemLockFree,KobayashietAlDeadlockFreeCalculus}, but our notion of usages
significantly differs from the original one, as discussed below.

\subsection{Indices}

First, we define integer indices, which are used
to keep track of the size of values in a process. 
%Those indices were for example used in \cite{DalLagoGaboardiLinearDependentTypes} and are greatly inspired by \cite{HughesParetoSabrySizedTypes}. 
%The main idea of those types in a sequential setting is to control recursive calls by ensuring a decrease in the sizes.
%In our setting, those sizes are useful to control replicated inputs. 

\begin{definition}
  	Let $\IV$ be a countable set of index variables, usually denoted by $i$,$j$ or $k$. The set of \emph{indices}, representing integers in $\NN_{\infty} = \NN \cup \set{\infty}$, is given by: % the following grammar. 
	$$ I,J := \INN{I} \midd \infty \qquad \INN{I} := i \midd f(\INN{I},\dots,\INN{I})  $$
        Here, $i \in \IV$. The symbol $f$ is an element of a given set of function symbols containing, for example, integers constants as nullary operators, addition and multiplication. We also assume the subtraction as a function symbol, with $n\minus m = 0$ when $m \ge n$. Each function symbol $f$ of arity $\mathtt{ar}(f)$ comes with an interpretation $\sem{}{f} : \NN^{\mathtt{ar}(f)} \rightarrow \NN $. 
	\label{d:indices}
\end{definition}
Given an index valuation $\rho : \IV \rightarrow \NN$, we extend the interpretation of function symbols to indices, noted $\sem{\rho}{I}$, as expected;
$\sem{\rho}{I}$ ranges over $\NN_{\infty}$.
%$$ \sem{\rho}{\infty} = \infty \qquad \sem{\rho}{i} = \rho(i) \qquad  \sem{\rho}{f(\INN{I},\dots,\INN{J})} = \sem{}{f}(\sem{\rho}{\INN{I}},\dots,\sem{\rho}{\INN{J}}) $$
For an index $I$, we
write $\Isub{I}{i}{\INN{J}}$ for the index obtained by replacing
the occurrences of $i$ in $I$ with $\INN{J}$.
Note that $\Isub{\infty}{i}{\INN{J}} = \infty$.

\begin{definition}[Constraints on Indices]
	Let $\varphi \subset \IV$ be a finite set of index variables. A \emph{constraint} $C$ on $\varphi$ is an expression with the shape $I \rel J$ where $I$ and $J$ are indices with free variables in $\varphi$ and $\rel$ denotes a binary relation on $\NNi$. Usually, we take $\rel \;\in \{ \le, <, =, \ne \}$. A finite set of constraints is denoted $\Phi$. 
	\label{d:constraints}
\end{definition}

For a finite set $\varphi \subset \IV$, we say that a valuation $\rho : \varphi \rightarrow \NN$ \emph{satisfies}
a constraint $I \rel J$ on $\varphi$, noted $\rho \q I \rel J$ when
$\sem{\rho}{I} \rel \sem{\rho}{J}$ holds.  
Similarly, $\rho \q \Phi$ holds when $\rho \q C$ for all $C \in \Phi$. Likewise, we note $\fifi \q C$ when for all valuations $\rho$ on $\varphi$ such that $\rho \q \Phi$ we have $\rho \q C$. We will also use some extended operations on indices $I,J$; for example, we may use $\infty + J = \infty$ or other such equations.

\subsection{Usages}

We use \emph{usages}
to express the channel-wise behaviour of a process.
%%Our notion of usages has been 
%%inspired by the usages introduced in
%%type systems for deadlock-freedom~\cite{Sumii98HLCL,KobayashiTypeSystemLockFree,KobayashietAlDeadlockFreeCalculus},
%%but differs from the original one in a significant manner.
Usages are a kind of CCS  processes  \cite{Milner89}
on a single channel, where each action is annotated with two time intervals. The set of usages, ranged over by $U$ and $V$, is given by:
$$ U,V ::= \ZU \mid (U|V) \mid \AU{\alpha}{A_o}{J_c}.U \mid~!U \mid U + V \qquad \alpha := \texttt{In} \mid \texttt{Out}$$ 
$$A_o,B_o ::= [I, J] \qquad J_c,I_c ::= J \mid [I,J] $$
Given a set of index variables $\varphi$ and a set of constraints $\Phi$, for an interval $[I,J]$, we always require that $\fifi \q I \le J$. For an interval $A_o = [I,J]$, we denote $\LeA(A_o) = I$ and $\RiA(A_o) = J$.
In the original notion of
  usages~\cite{Sumii98HLCL,KobayashiTypeSystemLockFree,KobayashietAlDeadlockFreeCalculus}, \(A_o\) and \(J_c\) were just numbers. The extension to intervals plays an important
  role in our analysis. Note that $J_c$ may be a single index $J$,
  but this single index $J$ should be understood as the interval $[-\infty,J]$.

Intuitively, a channel with usage $\ZU$ is not used at all. A channel of usage
  $U\PAR V$ can be used according to \(U\) and \(V\) possibly in parallel.
  The usage $\IU{A_o}{J_c}.U$ describes a channel that may be used for input,
  and then used according to \(U\). The two intervals \(A_o\) and \(J_c\),
  called \emph{obligation} and \emph{capacity} respectively, are used to achieve
  a kind of assume-guarantee reasoning. The obligation \(A_o\) indicates a \emph{guarantee} that if the channel is indeed used for input, then the input should become
  ready during the interval \(A_o\). The capacity \(J_c\) indicates the \emph{assumption}
  that if the environment performs a corresponding output, that output will be
  provided during the time interval \(J_c\) after the input becomes ready.
  For example, if a channel \(a\) has
  usage \(\IU{[1,1]}{J_c}.\ZU\),
  then the process \(\tick.a().0\) conforms to the usage,
  but \(a().0\) and \(\tick.\tick.a().0\) do not. 
  %Furthermore, if \(J_c=[0,1]\), and if a process \(\tick^k.\Out{a}{}\) is running in parallel with \(\tick.a().0\), then \(k\) belongs to the interval \([1,1]+[0,1] = [1,2]\). 
  %Furthermore, if \(J_c=[0,1]\), and if a process with an ouput on $a$ is running in parallel with \(\tick.a().0\), then this output must be ready in the interval \([1,1]+[0,1] = [1,2]\).
  Similarly, $\OU{A_o}{J_c}.U$ has the same meaning but for output. The usage $!U$ denotes the usage $U$ that can be replicated infinitely, and $U + V$ denotes a non-deterministic choice between the usages $U$ and $V$. This is useful for example in a case of pattern matching where a channel can be used very differently in the two branches. 

Recall that the obligation and capacity intervals in usages express a sort of
  assume-guarantee reasoning. We thus require that the assume-guarantee reasoning in a usage
  is ``consistent'' (or \emph{reliable}, in the terminology of usages). For example, the usage
  \(\IU{[0,0]}{[1,1]} \PAR \OU{[1,1]}{0}\) is reliable because (i)
  the part \(\IU{[0,0]}{[1,1]}\) assumes that a corresponding output will become ready
  at time \(1\), and the other part \(\OU{[1,1]}{0}\) indeed guarantees that and moreover, (ii)
  \(\OU{[1,1]}{0}\) assumes that a corresponding input will be ready by
  the time the output becomes ready, and the part \(\IU{[0,0]}{[1,1]}\) guarantees that.
  In contrast,  the usage \(\IU{[0,0]}{[1,1]} \PAR \OU{[2,2]}{0}\) is problematic
  because, although the part \(\IU{[0,0]}{[1,1]}\) assumes that an output will be ready at time \(1\),
  \(\OU{[2,2]}{0}\) provides the output only at time \(2\).
  The consistency on assume-guarantee reasoning must hold during the whole computation;
  for example, in the usage \( \IU{[0,0]}{[0,0]}.\IU{[0,0]}{[1,1]} \PAR
  \OU{[0,0]}{[0,0]}.\OU{[2,2]}{0}\), %%the assume/guarantee on
  the first input/output pair is fine,
  but the usage expressing the next  communication:
  \(\IU{[0,0]}{[1,1]} \PAR \OU{[2,2]}{0}\) is problematic.
  To properly define the reliability of usages during the whole computation, we first prepare
   a reduction semantics for usages, by viewing usages as CCS processes.

\begin{definition}[Congruence for Usages]
	The relation $\congr$ is defined as the least congruence relation closed under:
	$$ U \PAR \pzero \congr U \qquad U \PAR V \congr V \PAR U \qquad U \PAR (V \PAR W) \congr (U \PAR V) \PAR W $$
	$$ !\ZU \congr \ZU \qquad !U \congr ~!U \PAR U \qquad !(U \PAR V) \congr ~!U \PAR !V \qquad !!U \congr ~!U  $$
\end{definition}

%%We have the usual relations for parallel composition. We also add a relation defining replication $!U \congr !U \PAR U$. The other relations allow manipulation of replication. This will be useful later for the subusage relation, as it can increase the set of typable programs.

%%Now that we have congruence, as before, we give the reduction semantics. Let us
Before giving the reduction semantics, we introduce some notations.

\begin{definition}[Operations on Usages]
	We define the operations \(\IPlus\), \(\Ilub\), and \(+\) by:
	$$ \begin{array}{l}
		A_o \IPlus J = [0,\LeA(A_o) + J] \qquad A_o \IPlus [I,J] = [\RiA(A_o) + I,\LeA(A_o) + J] \\{}
		[I,J]\Ilub [I',J'] = [\max(I,I'), \max(J,J')] \qquad 
		[I,J]+ [I',J'] = [I+I', J+J'] 
	\end{array} $$
	Note that \(\IPlus\) is an operation that takes an obligation interval and a capacity and returns an interval. The operations \(\Ilub\) (max) and \(+\) are just pointwise
        extensions of the operations for indices.
        The intuition on \(\IPlus\) is explained later when we define the reduction
        relation.
%%        This is where we see the fact that a capacity $J$ alone should be understood as the interval $[-\infty,J]$. Indeed, when considering a single index, the lower bound of the sum becomes $0$. So, in particular, we have $A_o \IPlus J \ne A_o \IPlus [0,J]$ in general. A case where we need this single index will be explained in Example~\ref{e:semaphore2}. 
	
	The delaying operation \(\Up{A_o}U\) on usages is defined by:
	$$
	\begin{array}{l}
		\Up{A_o}\ZU = \ZU \qquad \Up{A_o}(U \PAR V) = \Up{A_o}U \PAR \Up{A_o}V \qquad \Up{A_o}(U + V) = \Up{A_o}U + \Up{A_o}V \\ \\ 
		\Up{A_o}\AU{\alpha}{B_o}{J_c}.U = \AU{\alpha}{A_o + B_o}{J_c}.U \qquad \Up{A_o}(!U) = !(\Up{A_o} U)
	\end{array}
	$$
	We also define $[I,J] + J_c$ and thus \(\Up{J_c}U\) by extending the operation with:
%	$$ [I,J] + J' = [I,J + J'] $$
	\( [I,J] + J' = [I,J + J'] \).
\end{definition}
Intuitively, a usage $\Up{A_o} U$ corresponds to the usage $U$ delayed
by a time approximated by the interval $A_o$. 
%Given two obligations $A_o$ and $B_o$, $A_o \Ilub B_o$ corresponds to an interval of time
%approximating the time for which those two obligations are
%respected. For example, if an input has the obligation to be ready in
%the interval of time $[4,8]$ and an output has the obligation to be
%ready in the interval of time $[5,7]$, then we know for sure that the
%input and the output will both be ready in the interval of time $[5,8]$.

The reduction relation is given by the rules of Figure~\ref{f:usagered}. 
%Intuitively,
The first rule means that to reduce a usage, we choose one input and
one output, and then we trigger the communication between them. This
communication occurs and does not lead to an error when the capacity
of an action indeed corresponds to a bound on the time the dual action is defined. This is given by the relation $A_o \subseteq B_o \IPlus J_c$. As an example, let us suppose that $B_o = [1,3]$, and the time for which the output becomes ready is in fact $2$, then the capacity $J_c$ says that after two units of time, the synchronization should happen in the interval $J_c$. So, if we take $J_c = [5,7]$ for example, then if $t$ is the time for which the dual input becomes ready, we must have $t \in [2 + 5,2 + 7]$. This should be true for any time value in $B_o$, so we want that $\forall t' \in [1,3], \forall t \in A_o, t \in [t' + 5, t' + 7]$, and this is equivalent to $A_o \subseteq B_o \IPlus [5,7] = [8,8]$. Indeed, $8$ is the only time that is in the three intervals $[6,8]$, $[7,9]$ and $[8,10]$. The case where $J_c = J$ is a single index occurs when $t$ can be smaller than $t'$, and in this case we only ask that the upper bound is correct: $\forall t' \in B_o, \forall t \in A_o, t \le t' + J$.

If the bound is incorrect, we trigger an error: see the second rule. In the case everything went well, the continuation is delayed by an approximation of the time when this communication occurs (see \(\Up{A_o\Ilub B_o}\) in the first rule).
%The idea is that we would like, after a communication, to synchronize the time of the continuation with the other subprocesses of a usage. As it is not easy to advance the time of all the other subprocesses, delaying the current subprocess leads to an easier semantics.
In the rules for $U + V$, a reduction step in usages can also make a non-deterministic choice. 

\begin{figure}
	\centering
	\begin{framed}
		\small 
		\begin{center}
			%synchronization
			\AXC{$\fifi \q B_o \subseteq A_o \IPlus I_c$}  
			\AXC{$\fifi \q A_o \subseteq B_o \IPlus J_c$}
			\BIC{$\fifi \p \IU{A_o}{I_c}.U \PAR \OU{B_o}{J_c}.V
				\ured \Up{A_o \Ilub B_o}(U \PAR V)$} 
			\DP 
			\qquad 
			%synchronization Error
			\AXC{$
				\fifi \nq (B_o \subseteq A_o \IPlus I_c \land A_o \subseteq B_o \IPlus J_c)$}
			\UIC{$\fifi \p \IU{A_o}{I_c}.U \PAR \OU{B_o}{J_c}.V
				\ured \Error$}
			\DP 
			\\
			\vvskip 
			%Choice 1
			\AXC{}
			\UIC{$\fifi \p U + V \ured U$}
			\DP
			\qquad  
			%Choice 2
			\AXC{}
			\UIC{$\fifi \p U + V \ured V$}
			\DP
			\qquad 
			%Parallel Context
			\AXC{$\fifi \p U \ured U' \qquad U' \ne \Error$}
			\UIC{$\fifi \p U \PAR V \ured U'\PAR V$}
			\DP 
			\\ 
			\vvskip 
			%Propagation of error
			\AXC{$\fifi \p U \ured \Error$}
			\UIC{$\fifi \p U \PAR V \ured \Error$}
			\DP
			\qquad 
			%Congruence
			\AXC{$U \equiv U' \qquad \fifi \p U'\ured V' \qquad V'\equiv V$}
			\UIC{$\fifi \p U \ured V$}
			\DP 
		\end{center}   
	\end{framed}
	\caption{Reduction Rules for Usages}
	\label{f:usagered}
\end{figure}

%Intuitively,
%%when a usage can reduce to an error at some point, it
An error in a usage reduction means that the assume-guarantee reasoning was inconsistent. Based on this intuition,
we define the notion of reliablity.

\begin{definition}[Reliability]
  A usage $U$ is \emph{reliable} under $\fifi$
  if \(U \not\ured^* \Error\).
\end{definition}

\begin{example} 
Consider the usage 
\(U := \IU{[1,1]}{1}. \OU{[1,1]}{0} \PAR \IU{[1,1]}{1}. \OU{[1,1]}{0} \PAR \OU{[0,0]}{[1,1]} \).
The only possible reduction sequence (with symmetry) is:
\[ U \ured \OU{[2,2]}{0} \PAR \IU{[1,1]}{1}. \OU{[1,1]}{0}\ured \OU{[3,3]}{0}.\]
For the first step, we have indeed
$[1,1] \subseteq [0,0] \IPlus [1,1] = [1,1]$ and $[0,0] \subseteq [1,1] \IPlus 1 = [0,2]$. Note that the capacity $[0,1]$ instead of $1$ for the input would not have worked since $[1,1] \IPlus [0,1] = [1,2]$.
%%Then, we end the reduction with 
%%$$ \OU{[2,2]}{0} \PAR \IU{[1,1]}{1}. \OU{[1,1]}{0} $$
%%since $[2,2] \subseteq [1,1] \IPlus 1 = [0,2]$ and $[1,1] \subseteq [2,2] \IPlus 0 = [0,2]$.
Thus, the usage \(U\) is reliable. It corresponds, for example,
to the usage of the channel $a$ in the process $P$ given in Example~\ref{e:semaphore}:
\( \tick. a(). \tick. \overline{a} \langle \rangle \PAR \tick. a(). \tick. \overline{a} \langle \rangle \PAR \overline{a} \langle \rangle  \).
The obligation $[1,1]$ corresponds to waiting for exactly one tick.
Then, the capacities say that once they are ready, the two inputs will
indeed communicate before one time unit for any reduction. And at the
end, we obtain an output available at time $3$, and this output has no
communication. One can see that those capacities and obligations indeed give the complexity of this process. Thus, we will ask in the type system that all usages are reliable, and so the time indications will give some complexity bounds on the behaviour of a channel. \qed
\label{e:semaphore2}
\end{example}

\begin{example}
  We give an example of a non-reliable usage.
  To the previous example, let us add another input in parallel
\small 
$$U := \IU{[1,1]}{1}. \OU{[1,1]}{0} \PAR \IU{[1,1]}{1}. \OU{[1,1]}{0} \PAR \IU{[1,1]}{1}. \OU{[1,1]}{0} \PAR \OU{[0,0]}{[1,1]} $$ 
\normalsize 
We have:
\( U \ured^* \OU{[3,3]}{0} \PAR \IU{[1,1]}{1}. \OU{[1,1]}{0} \ured \Error \),
because $[1,1] \IPlus 1 = [0,2]$, so the capacity here is not a good assumption.
%%Therefore, this usage is not reliable.
However, the following variant of the usage:
\small 
$$ U := \IU{[1,1]}{2}. \OU{[1,1]}{0} \PAR \IU{[1,1]}{2}. \OU{[1,1]}{0} \PAR \IU{[1,1]}{2}. \OU{[1,1]}{0} \PAR \OU{[0,0]}{[1,1]} $$
\normalsize 
is reliable. This example shows how reliability adapts to parallel composition. \qed
\label{e:semaphorefalse}
\end{example} 

We introduce another relation \(U\SubU V\) called the \emph{subusage} relation, which will be used later to define the subtyping relation. It is defined by the rules of Figure~\ref{f:subusage}. The relation $U \SubU V$ intuitively means that any channel of usage $U$ may also be used according to $V$. For example, $U \SubU 0$ says that we may not use a channel (usage equal to $0$). Recall that an obligation and a capacity express a guarantee and an assumption respectively.
The last but one rule says
that it is safe to strengthen the guarantee and weaken the
assumption. We use the relation $I_c \le J_c$ to denote the relation $\subseteq$ on intervals, where a single index $J$ is considered as the interval $[-\infty,J]$. 
The last rule can be understood as follows.
The part \(\Up{A_o + J_c}V\) says that a channel may be used according to
\(V\) only after the interval \(A_o+J_c\). Since the action \(\alpha^{A_o}_{J_c}\)
is indeed finished during the interval \(A_o+J_c\), we can move \(V\) to
under the guard of \(\alpha^{A_o}_{J_c}\). This last rule is especially useful for substitution, as explained in the example below.

\begin{example} 
	Consider the process:
	 $$P := \In{a}{r}. \In{r}{}. \In{b}{} \PAR \Out{a}{b} $$	 
	Let us give usages to \(b\) and \(r\); here we omit time annotations for the sake of simplicity. We have $U_r = \IU{}{}$ and $U_b = \IU{}{} \PAR U_r$
	Indeed, $r$ is used only once as an input, and $b$ is used as an input on the left, and it is sent to be used as $r$ on the right. Thus, after a reduction step we obtain $P \red \In{b}{}. \In{b}{}$ where $b$ has usage $U_b' = \IU{}{}. \IU{}{}$.
	So, the channel \(b\) had usage \(U_b\) in \(P\), but it ended up being used according to
	\(U_b'\); that is valid since we have the subusage relation \(U_b\SubU U_b'\).
	\label{e:substitution}
\end{example}

\begin{figure*}
	\centering
	\begin{framed}
		\begin{center}
			%Erasing 
			\AXC{} 
			\UIC{$\fifi \p U \SubU \ZU$}
			\DP 
			\qquad
			%Choice
			\AXC{$i \in \set{1;2}$} 
			\UIC{$\fifi \p U_1 + U_2 \SubU U_i$}
			\DP 
			\qquad 
			%Choice Context Left 
			\AXC{$\fifi \p U \SubU U'$}
			\UIC{$\fifi \p U + V \SubU U' + V$}
			\DP 
			\\ 
			\vvskip   
			%Choice Context Right
			\AXC{$\fifi \p V \SubU V'$}
			\UIC{$\fifi \p U + V \SubU U + V'$}
			\DP 
			\qquad 
			%Parallel Context
			\AXC{$\fifi \p U \SubU U'$}
			\UIC{$\fifi \p U \PAR V \SubU U' \PAR V$}
			\DP  
			\qquad 
			%Replication Context
			\AXC{$\fifi \p U \SubU U'$}
			\UIC{$\fifi \p !U  \SubU !U' $}
			\DP
			\\ 
			\vvskip  
			%Congruence
			\AXC{$U \equiv U' \qquad \fifi \p U' \SubU V' \qquad V \equiv V'$}
			\UIC{$\fifi \p U \SubU V$}
			\DP 
			\qquad  
			%Transitivity
			\AXC{$\fifi \p U \SubU U' \qquad \fifi \p U' \SubU U''$}
			\UIC{$\fifi \p U \SubU U''$}
			\DP 
			\\ 
			\vvskip   
			%Action Context
			\AXC{$\fifi \p U \SubU U'$}
			\UIC{$\fifi \p \alpha^{A_o}_{J_c}. U \SubU \alpha^{A_o}_{J_c}. U'$} 
			\DP 
			\qquad  
			%Action Rule
			\AXC{$\fifi \q B_o \subseteq A_o \qquad \fifi \q I_c \le J_c$}
			\UIC{$\fifi \p \alpha^{A_o}_{I_c}. U \SubU \alpha^{B_o}_{J_c}. U$}
			\DP 
			\\ 
			\vvskip 
			\AXC{}
			%Usual Scope Rule
			\UIC{$\fifi \p (\alpha^{A_o}_{J_c} . U) \PAR (\Up{A_o + J_c}V) \SubU \alpha^{A_o}_{J_c}. (U \PAR V) $}
			\DP 
		\end{center}   
	\end{framed}
	\caption{Subusage}
	\label{f:subusage}
\end{figure*}

%Before continuing to the type system, we give some intermediate results on subusages. 

\subsection{Type System} 

We extend ordinary types for the \(\pi\)-calculus with usages.
\begin{definition}[Usage Types]
	We define \emph{types} by the following grammar:
	$$ T,S ::= \Nat[I,J] \midd \ChT{\seq{T}}{U} \midd \ServT{\seq{i}}{K}{\seq{T}}{U}.  $$
\end{definition} 
The type \(\Nat[I,J]\) describes an integer \(n\) such that \(I\le n\le J\).
Channels are classified into \emph{server channels} (or just \emph{servers})
  and \emph{simple channels}.
  All the inputs on a server channel must be
  replicated (as in $!\In{a}{\seq{\var}}.P$), while
  no input on a simple channel can be replicated.
  The type  \(\ChT{\seq{T}}{U}\) describes a simple channel
  that is used for transmitting values of type \(\seq{T}\) according to usage \(U\).
  For example, $\ChT{\Nat[I,J]}{U}$ is the type of channels used according to \(U\)
  for transmitting integers in the interval \([I,J]\).
  The type \(\ServT{\seq{i}}{K}{\seq{T}}{U}\) describes a server channel
  that is used for transmitting values of type \(\seq{T}\) according to usage \(U\);
  the superscript \(K\), which we call the \emph{complexity} of a server, is an interval. It 
  denotes the cost incurred when a server is invoked.
  Note that the server type allows polymorphism on index variables \(\seq{i}\).

  The subtyping relation \(T\subtype T'\), which means
  that a value of type \(T\) can also be used as
  a value of type \(T'\),
  is defined by the rules of Figure~\ref{f:usagesubtype}.

\begin{figure*}
	\centering
	\begin{framed}
		\small 
		\begin{center}
			\AXC{$\fifi \q I' \le I$}
			\AXC{$\fifi \q J \le J'$}
			\BIC{$\fifi \p \Nat[I,J] \subtype \Nat[I',J']$}
			\DP  
			\qquad 
			\AXC{$\fifi \p \seq{T} \subtype \seq{T'}$}
			\AXC{$\fifi \p \seq{T'} \subtype \seq{T}$}
			\AXC{$\fifi \p U \SubU V$}
			\TIC{$\fifi \p \ChT{\seq{T}}{U} \subtype \ChT{\seq{T'}}{V}$}
			\DP 
			\\ 
			\vvskip 
			\AXC{$\varphi,\seq{i};\Phi \p \seq{T} \subtype \seq{T'}$}
			\AXC{$\varphi,\seq{i};\Phi \p \seq{T'} \subtype \seq{T}$}
			\AXC{$\varphi,\seq{i};\Phi \q K = K'$}
			\AXC{$\fifi \p U \SubU V$}
			\QIC{$\fifi \p \ServT{\seq{i}}{K}{\seq{T}}{U} \subtype \ServT{\seq{i}}{K'}{\seq{T'}}{V}$}
			\DP 
		\end{center}   
	\end{framed}
	\caption{Subtyping Rules for Usage Types}
	\label{f:usagesubtype}
\end{figure*}

We extend operations on usages to partial operations on types and typing contexts with $\Gamma = \var_1 : T_1, \dots, \var_n : T_n$. The delaying of a type $\Up{A_o} T$ is defined as the delaying of the usage for a channel or a server type, and it does nothing on integers.
%The delaying operation for a context delays all the types in the context.
%%Then delaying a context correspond to delaying all types in this context.
We also say that a type is \emph{reliable} when it is an integer type, or when it is a server or channel type with a reliable usage. We define following operations: 

\begin{definition} %[Parallel Composition of Types]
	The parallel composition $T\,|\, T'$ is defined by:
	\small 
	$$ \Nat[I,J] \PAR \Nat[I,J] = \Nat[I,J] \qquad  \ChT{\seq{T}}{U} \PAR \ChT{\seq{T}}{V} = \ChT{\seq{T}}{(U \PAR V)} $$ 
	$$\ServT{\seq{i}}{K}{\seq{T}}{U} \PAR \ServT{\seq{i}}{K}{\seq{T}}{V} = \ServT{\seq{i}}{K}{\seq{T}}{(U \PAR V)} $$ 
	\normalsize 
	
\end{definition}

\begin{definition}[Replication of Type]
	The replication of a type $!T$ is defined by:
	\small 
	$$ !\Nat[I,J] = \Nat[I,J] \qquad !\ChT{\seq{T}}{U} = \ChT{\seq{T}}{(!U)} \qquad  !\ServT{\seq{i}}{K}{\seq{T}}{U} = \ServT{\seq{i}}{K}{\seq{T}}{(!U)} $$ 
	\normalsize 
\end{definition}

The (partial) operations on types defined above are extended pointwise to contexts.
For example, for $\Gamma = \var_1 : T_1, \dots, \var_n : T_n$ and $\Delta = \var_1 : T_1', \dots, \var_n : T_n'$, we define $\Gamma \PAR \Delta = \var_1 : T_1 \PAR T_1', \dots, \var_n : T_n \PAR T_n'$. Note that this is defined just if $\Gamma$ and $\Delta$ agree on the typing of integers and associate the same types (excluding usage) to names.  

%%We also introduce the following notation.

\begin{definition}\label{def:seqcomp}
	Given a capacity $J_c$ and an interval $K = [K_1,K_2]$, we define $J_c;K$ by; 
	$$ J;[K_1,K_2] = [0,J + K_2] \qquad [\infty,\infty];[K_1,K_2] = [0,0] \qquad [\INN{I},J];[K_1,K_2] = [0,J + K_2] $$
\end{definition} 
\noindent
Intuitively, \(J_c;K\) represents the complexity of an input/output process when the input/output has capacity \(J_c\) and the complexity of the continuation is \(K\). \(J_c=[\infty,\infty]\) means the input/output will never succeed (because there is no corresponding output/input); hence
the complexity is \(0\). A case where this is useful is given later in Example~\ref{e:deadlock}. Otherwise, an upper-bound is given by \(J+K_2\) (the time spent for the input/output to succeed, plus \(K_2\)). The lower-bound is \(0\), since the input/output may be blocked forever.

The type system is given in Figures~\ref{f:sizetypeexpression} and %Figure~
\ref{f:usagetype}. The typing rules for expressions are standard ones for sized types.
\begin{figure}[tbp]
	\centering
	\begin{framed}
		\begin{center}
			\AXC{$\var \COL T \in \Gamma$}
			\UIC{$\fifi;\Gamma \p \var \COL T$}
			\DP 
			\qquad
			\AXC{}
			\UIC{$\fifi; \Gamma \p \zero \COL \Nat[0,0]$}
			\DP 
			\qquad  
			\AXC{$\fifi; \Gamma \p e \COL \Nat[I,J]$}
			\UIC{$\fifi; \Gamma \p \suc(e) \COL \Nat[I+1,J+1]$}
			\DP 
			\\ 
			\vvskip   
	\AXC{$\fifi; \Delta \p e \COL T' $}
%			\AXC{$\fifi; \Delta \p e \COL U$}
			\AXC{$\fifi \p \Gamma \subtype \Delta$}
			\AXC{$\fifi \p 	T' \subtype T$}
%	\AXC{$\fifi \p U \subtype T$}
			\TIC{$\fifi;\Gamma \p e \COL T$}
			\DP 
		\end{center}   
	\end{framed}
	\caption{Typing Rules for Expressions}
	\label{f:sizetypeexpression}
\end{figure}

A type judgment for processes is of the form $\fifi;\Gamma \p P \C [I,J]$
  where $\varphi$ denotes the set of index variables,
  $\Phi$ is a set of constraints on index variables, and
  $J$ is a bound on the parallel complexity of $P$ under those constraints. This complexity bound $J$ can also be seen as a bound on the open complexity of a process, that is to say the complexity of $P$ in an environment corresponding to the types in $\Gamma$. For example, a channel with usage $\IU{[1,1]}{5}$ alone cannot be reduced, as it is only used as an input. So, the typing $\void;a : \ChT{}{\IU{[1,1]}{5}} \p \tick. \In{a}{} \C [1,6]$ says that in an environment that may provide an output on the channel $a$ within the time interval $[1,1] \IPlus 5 = [0,6]$, this process has a complexity bounded by $6$. Similarly, the lower bound $I$ is a lower bound on the parallel complexity of $P$. But in practice, this lower bound is often too imprecise.\footnote{This is because in the definition of
    \(J_C;K\) in Definition~\ref{def:seqcomp}, we pessimistically take into account
    the possibility that each input/output may be blocked forever. We can
    avoid the pessimistic estimation of the lower-bound
    by incorporating information about lock-freedom~\cite{KobayashiTypeSystemLockFree,KobayashiTypeBasedInformationFlowAnalysis}.}

\begin{figure}[tbp]
	\centering
	\begin{framed}
		\footnotesize  
		\begin{center}
			%Zero
			\AXC{}
			\LeftLabel{(zero)}
			\UIC{$\fifi;\Gamma \p \pzero \C [0,0]$}
			\DP 
			\qquad
			%Parallel Composition 
			\AXC{$\fifi;\Gamma \p P \C K_1$}
			\AXC{$\fifi;\Delta \p Q \C K_2$}
			\LeftLabel{(par)}
			\BIC{$\fifi;\Gamma \PAR \Delta \p P \PAR Q \C K_1 \Ilub K_2$}
			\DP 
			\\ 
			\vvskip 
			%Tick
			\AXC{$\fifi;\Gamma \p P \C K$}
			\LeftLabel{(tick)}
			\UIC{$\fifi;\Up{[1,1]}{\Gamma} \p \tick.P \C K + [1,1]$}
			\DP 
			\qquad 
			%Channel Input  			
			\AXC{$\fifi;\Gamma, a \COL \ChT{\seq{T}}{U}, \seq{\var} \COL \seq{T} \p P \C K$}
			\LeftLabel{(ich)}
			\UIC{$\fifi;\Up{J_c}{\Gamma}, a \COL \ChT{\seq{T}}{\IU{[0,0]}{J_c}.U} \p \In{a}{\seq{\var}}.P \C J_c;K$}
			\DP 
			\\ 
			\vvskip 
			%Server Input
			\AXC{$(\varphi,\seq{i});\Phi;\Gamma, a \COL \ServT{\seq{i}}{K}{\seq{T}}{U}, \seq{\var} \COL \seq{T} \p P \C K$}
			\LeftLabel{(iserv)}
			\UIC{$\fifi;\Up{J_c}{!\Gamma}, a \COL \ServT{\seq{i}}{K}{\seq{T}}{!\IU{[0,0]}{J_c}.U} \p !\In{a}{\seq{\var}}.P \C [0,0]$}
			\DP
			\\ 
			\vvskip 
			%Channel Output 
			\AXC{$\fifi;\Gamma', a \COL \ChT{\seq{T}}{V}  \p \seq{e} \COL \seq{T} \qquad \fifi;\Gamma, a \COL \ChT{\seq{T}}{U} \p P \C K$}
			\LeftLabel{(och)}
			\UIC{$\fifi;\Up{J_c}{(\Gamma \PAR \Gamma')}, a \COL \ChT{\seq{T}}{\OU{[0,0]}{J_c}.(V \PAR U)} \p \Out{a}{\seq{e}}.P \C J_c;K$}
			\DP 
			\\ 
			\vvskip 
			%Server Output
			\AXC{$\fifi;\Gamma', a \COL \ServT{\seq{i}}{K}{\seq{T}}{V}  \p \seq{e} \COL \Isub{\seq{T}}{\seq{i}}{\seq{\INN{I}}} \qquad \fifi;\Gamma, a \COL \ServT{\seq{i}}{K}{\seq{T}}{U} \p P \C K' $}
			\LeftLabel{(oserv)}
			\UIC{$\fifi;\Up{J_c}{(\Gamma \PAR \Gamma')}, a \COL \ServT{\seq{i}}{K}{\seq{T}}{\OU{[0,0]}{J_c}.(V \PAR U)} \p \Out{a}{\seq{e}}.P \C J_c;(K' \Ilub \Isub{K}{\seq{i}}{\seq{\INN{I}}})$}
			\DP 
			\\ 
			\vvskip
			%Pattern Matching
			\AXC{$\fifi;\Gamma \p e \COL \Nat[I,J]$}
			\AXC{$\varphi;\Phi,I \le 0;\Gamma \p P \C K$}
			\AXC{$\varphi;\Phi,J \ge 1;\Gamma, x \COL \Nat[I \minus 1,J \minus 1] \p Q \C K$}
			\LeftLabel{(if)}
			\TIC{$\fifi;\Gamma \p \pifn{e}{P}{x}{Q} \C K$}
			\DP 
			\\
			\vvskip  
			%Name Creation
			\AXC{$\fifi;\Gamma, a \COL T \p P \C K \qquad T \text{ reliable}$}
			\LeftLabel{(nu)}
			\UIC{$\fifi;\Gamma \p (\nu a) P \C K$}
			\DP 
			\\ 
			\vvskip 
			%Subtyping rule
			\AXC{$\fifi;\Delta \p P \C K$}
			\AXC{$\fifi \p \Gamma \subtype \Delta$}
			\AXC{$\fifi \q K \subseteq K'$}
			\LeftLabel{(subtype)}
			\TIC{$\fifi;\Gamma \p P \C K'$}
			\DP
		\end{center}   
	\end{framed}
	\caption{Typing Rules for Processes}
	\label{f:usagetype}
\end{figure}
  
The (par) rule separates a context into two parts, and the complexity is the maximum over the two complexities, both for lower bound and upper bound. 
The (tick) rule shows the addition of a tick implies a delay of $[1,1]$ in the context and the complexity. 
The (nu) rule imposes that all names must have a reliable usage when they are created. 
In order to type a channel with the (ich) rule, the channel must have an input usage, with obligation $[0,0]$. Note that with the subusage relation, we have $\IU{A_o}{J_c} \SubU \IU{[0,0]}{J_c}$ if and only if $A_o = [0,I]$ for some $I$. 
So, this typing rule imposes that the lower-bound guarantee is correct, but the rule is not restrictive for upper-bound. 
This rule induces a delay of $J_c$ in both context and complexity. 
Indeed, in practice this input does not happen immediately as we need to wait for output. 
This is where the assumption on when this output is ready, given by the capacity, is useful. 
The rule for output (och) is similar. 
For a server, the rule for input (iserv) is similar to (ich) in principle but differs in the way complexity is managed. 
Indeed, as a replicated input is never modified nor erased through a computation, giving it a non-zero complexity would harm the precision of the type system. 
Moreover, if this server represents for example a function on an integer with linear complexity, then the complexity of this server depends on the size of the integer it receives; that is why the complexity is transferred to the output rule on server, as one can see in the rule (oserv). 
Indeed, this rule (oserv) is again similar to (och) but the complexity of a call to the server is added in the rule. 
As we have polymorphism on servers, in order to type an output we need to find an instantiation on the indices $\seq{i}$, which is denoted by $\seq{\INN{I}}$ in this rule. 
Finally, the (if) rule is the only rule that modifies the set of constraints, and it gives information on the values the sizes can take. 
As explained in Example~\ref{e:factorial}, those constraints are crucial in our sized type system. 
Note that contexts are not separated in this rule. 
%In the typing for the expression this is not a problem since names are not useful for the typing of an integer.
So, for both branches, it means that the usage of channels must be the same. 
However, because we have the choice usage ($U + V$), in practice we can use different usages in those two branches. 

%%As an illustration of the type system, let us take back again the reliable usage described in Example~\ref{e:semaphore2} and show that it corresponds indeed to the typing of $a$ in the process described in Example~\ref{e:semaphore}.

\begin{example}
	The typing derivation of the process in Example~\ref{e:semaphore} is given in Figure~\ref{f:typingsemaphore}. Note that the process 
	%\small 
	$( \tick. \In{a}{}. \tick. \Out{a}{} \PAR \tick. \In{a}{}. \tick. \Out{a}{} \PAR \tick. \In{a}{}. \tick. \Out{a}{} \PAR \Out{a}{} \)
	\normalsize 
	is also typable, in the same way, using the
        following usage (recall Example~\ref{e:semaphorefalse}).
        
	\small 
	$$ U := \IU{[1,1]}{2}. \OU{[1,1]}{0} \PAR \IU{[1,1]}{2}. \OU{[1,1]}{0} \PAR \IU{[1,1]}{2}. \OU{[1,1]}{0} \PAR \OU{[0,0]}{[1,1]} $$
	\normalsize
	We thus obtain the complexity bound $[1,4]$.
	
	\begin{figure*}
		\footnotesize 
		\begin{framed}
			\begin{prooftree} 	
				\AXC{}
				\UIC{$\void;a : \ChT{}{\OU{[0,0]}{0}} \p \Out{a}{} \C [0,0]$}
				\UIC{$\void;a : \ChT{}{\OU{[1,1]}{0}} \p \tick. \Out{a}{} \C [1,1]$}
				\UIC{$\void;a : \ChT{}{(\IU{[0,0]}{1}. \OU{[1,1]}{0})} \p \In{a}{}. \tick. \Out{a}{} \C [0,2]$}
				\UIC{$\void;a : \ChT{}{(\IU{[1,1]}{1}. \OU{[1,1]}{0})} \p \tick. \In{a}{}. \tick. \Out{a}{} \C [1,3]$} 
				\AXC{}
				\UIC{$\void;a : \ChT{}{(\OU{[0,0]}{[1,1]})} \p \Out{a}{} \C [0,1]$} 
				\BIC{$\void;a : \ChT{}{(\IU{[1,1]}{1}. \OU{[1,1]}{0} \PAR \IU{[1,1]}{1}. \OU{[1,1]}{0} \PAR \OU{[0,0]}{[1,1]})} \p \tick. \In{a}{}. \tick. \Out{a}{} \PAR \cdots \PAR \Out{a}{} \C [1,3]$}
			\end{prooftree}
		\end{framed} 
		\caption{Typing of Example~\ref{e:semaphore}}
		\label{f:typingsemaphore}
	\end{figure*}
	
\end{example} 
An example for the use of servers and sizes is given later, in Example~\ref{e:factorial}, as well as a justification for the use of intervals for obligations and capacities, in Example~\ref{e:intervals}. 
%\begin{remark}
%	A careful reader may wonder why we need \emph{intervals} for
%obligations and capacities, instead of single numbers. An informal justification is given in the Appendix~\ref{ss:intervals}.
%\end{remark} 

\section{Soundness and Examples} 
\label{s:soundness}

%The proof of soundness relies on standard lemmas for type systems, mainly substitution lemmas and subject reduction.
The proof of soundness relies on subject reduction. In order to work on the parallel reduction relation $\pred$, we need to consider annotated processes. 
%So, in the following we will always consider $P$ as an annotated process. 
We introduce the following typing rule, for the annotation, as a generalization of the rule for $\tick$.
\begin{prooftree}
	\AXC{$\fifi;\Gamma \p P \C K$}
	\UIC{$\fifi;\Up{[m,m]}{\Gamma} \p m : P \C K + [m,m]$}
\end{prooftree}  

\subsection{Subject Reduction and Soundness}
In the Appendix~\ref{ss:intermediate}, we describe some intermediate lemmas needed for the soundness proof, namely weakening, strengthening and then substitution lemmas for index and expressions.  

Let us first introduce a notation for subject reduction:
\begin{definition}[Reduction for Contexts]
	We say that a context $\Gamma$ reduces to a context $\Gamma'$ under $\fifi$, denoted $\fifi \p \Gamma \ured^* \Gamma'$ when one of the following holds:
	\begin{itemize}
%		\item $\Gamma = \Gamma'$
		\item $\Gamma = \Delta, a : \ChT{\seq{T}}{U} \qquad \fifi \p U \ured^* U'\qquad \Gamma' = \Delta, a : \ChT{\seq{T}}{U'}$ 
		\item $\Gamma = \Delta, a : \ServT{\seq{i}}{K}{\seq{T}}{U} \qquad \fifi \p U \ured^* U' \qquad \Gamma' = \Delta, a : \ServT{\seq{i}}{K}{\seq{T}}{U'}$
	\end{itemize}
\end{definition}

So, $\Gamma'$ is $\Gamma$ after some reduction steps but only in a unique usage. We obtain immediately that if all types in $\Gamma$ are reliable then all types in $\Gamma'$ are also reliable by definition of reliability. 

The subject reduction property is stated as follows;
see Appendix~\ref{ss:subjectreduction} for a proof.

\begin{theorem}[Subject Reduction]
	If $\fifi; \Gamma \p P \C K$ with all types in $\Gamma$ reliable and $P \pred Q$ then there exists $\Gamma'$ with $\fifi \p \Gamma \ured^* \Gamma'$ and $\fifi; \Gamma' \p Q \C K$. 
	\label{t:parallelsubjectreduction}
\end{theorem}

\begin{comment}
In order to do that, we need first a lemma saying that the congruence relation behaves well with typing. 

\begin{lemma}[Congruence and Typing]
	Let $P$ and $Q$ be annotated processes such that $P \congr Q$. Then, $\fifi; \Gamma \p P \C K$ if and only if $\fifi; \Gamma \p Q \C K$.
	\label{l:parallelcongrtype}
\end{lemma}

This is proved by induction on $P \congr Q$, and it relies on the
congruence in usages and in the cases where we modify annotations, it
mainly relies on the fact that delaying does not modify the behaviour of a type, as expressed by Lemma~\ref{l:delayinginvariance}. Proof elements are given in the Appendix~\ref{ss:congruence}.

And now that we can work up to the congruence relation with
Lemma~\ref{l:parallelcongrtype}, Theorem~\ref{t:parallelsubjectreduction}
is proved by induction on $P \pred Q$. Without surprise, the most
difficult case is for a communication, and it greatly relies on reliability, see Appendix~\ref{ss:subjectreduction} for details. 
\end{comment}
The following is the main soundness theorem.
\begin{theorem}
	Let $P$ be an annotated process and $n$ be its global parallel
        complexity. Then, if $\fifi;\Gamma \p P \C [I,J]$ with all types in $\Gamma$ reliable, then we have $\fifi \q J \ge n$. Moreover, if $\Gamma$ does not contain any integers variables, we have $\fifi \q I \le n$. 
	\label{t:mainparallel}
\end{theorem}

\begin{proof}
%	\AG{Maybe the proof can be moved to the appendix, or significantly shortened}
	By Theorem~\ref{t:parallelsubjectreduction}, all reductions from $P$ using $\pred$ conserve the typing. The context may be reduced too, but as a reduction step does not harm reliability, we can still apply the subject reduction through all the reduction steps of $\pred$. Moreover, for a process $Q$, if we have a typing $\fifi;\Gamma \p Q \C [I,J]$, then $J \ge \lcomp(Q)$.
	 %Indeed, a constructor $n : P$ forces an increment of the complexity of $n$ both in typing and in the definition of $\lcomp(Q)$, and for parallel composition the typing imposes a complexity greater than the maximum as in the definition for $\lcomp(Q)$. 
	 Thus, $J$ is indeed a bound on the parallel complexity by
         definition. As for the lower bound, one can see that we do
         not always have $I \le \lcomp(Q)$ because of the processes $\tick. Q'$ and $\pifn{e}{Q_1}{x}{Q_2}$. However, those two processes are not in normal form for $\pred$, because $\tick. Q' \pred 1 : Q'$ and as there are no integer variables in $\Gamma$, the pattern matching can also be reduced. Thus, from a process $Q$ we can find $Q'$ such that $Q \pred Q'$ and $Q'$ has no such processes of this shape on the top. And then, we obtain $I \le \lcomp(Q')$ which is smaller than the parallel complexity of $Q$ by definition. 
\end{proof}

\subsection{Examples} 

%Let us also present how sizes and polymorphism over indices in servers can type processes defined by replication such as the factorial. Please note that by taking inspiration from the typing in \cite{BaillotGhyselen21}, using the type representation given in the Appendix~\ref{ss:comparison}, more complicated examples of parallel programs such as the bitonic sort could be typed in our setting with a good complexity bound.  
Below we give several examples to demonstrate the expressive power of our type system. 

\begin{example}[Intervals] 
	\label{e:intervals}
	
	To see the need for an interval capacity, consider the following process: 
	$$ \In{a}{}.\Out{b}{} \PAR \pifn{e}{\Out{a}{}}{x}{\tick.\Out{a}{}} $$
	Depending on the value of \(e\) (which may be statically unknown),
	an output on \(a\) may be available at time 0 or 1. Thus, the input usage on \(a\) should have
	a capacity interval \([0,1]\). As a result, the obligation of the
	output usage on \(b\) should also be an interval \([0,1]\).
	
	Now, one may think that we can assume that lower-bounds are always \(0\) and omit lower-bounds,
	since we are mainly interested in an \emph{upper-bound} of the parallel complexity.
	Information about lower-bounds is, however, actually required for precise reasoning on
	upper-bounds. For example, consider the process 
	$$\In{a}{}.\Out{b}{} \PAR \tick.\Out{a}{}.\In{b}{}$$
	With intervals, \(a\) have the usage
	\(\IU{[0,0]}{[1,1]}\PAR \OU{[1,1]}{0}\)
	and so $b$ has the usage
	\(\OU{[1,1]}{[0,0]}\PAR \IU{[1,1]}{[0,0]}\), and the parallel complexity of
	the process can be precisely inferred to be \(1\).
	
	If we set lower-bounds to \(0\) and assign to $a$ the usage 
	\(\IU{[0,0]}{[0,1]}\PAR \OU{[0,1]}{0}\),
	then the usage of \(b\) can only be:
	\(\OU{[0,1]}{\textcolor{red}{1}}\PAR \IU{[0,1]}{\textcolor{red}{1}}\).
	Note that according to the imprecise usage of \(a\),
	the output on \(b\) may become ready at time \(0\) and then have to wait for one time unit
	until the input on \(b\) becomes ready;
	thus, the capacity of the output on \(b\) is \(1\), instead of \([0,0]\).
	An upper-bound of the parallel complexity would therefore be inferred to be \(1+1=2\) (because the usages tell us that
	the lefthand side process may wait for one time unit at \(a\),
	and then for another time unit at \(b\)), which is too imprecise.
	
	We remark that this problem does not come for an inappropriate definitions of usages with only upper-bound in our work. Indeed, by adapting the usage type system given in \cite{KobayashiTypeSystemLockFree}, we would have the same imprecision. In the same way, trying to give a notion of reliability that makes the usage
	\(\OU{[0,1]}{\textcolor{red}{0}}\PAR \IU{[0,1]}{\textcolor{red}{0}}\)
	reliable would lead to an unsound type system, as it would make the subusage relation less flexible, which is essential for soundness.
	\qed
\end{example} 

Let us also present how sizes and polymorphism over indices in servers can type processes defined by replication such as the factorial. Please note that by taking inspiration from the typing in \cite{BaillotGhyselen21}, using the type representation given in the Appendix~\ref{ss:comparison}, more complicated examples of parallel programs such as the bitonic sort could be typed in our setting with a good complexity bound. 

\begin{example}[Factorial]\label{e:factorial}
  Assume a function on expressions $\mathtt{mult} : \Nat[I,J] \times \Nat[I',J'] \rightarrow \Nat[I * I',J * J']$. In practice, this should be encoded as a server in $\pi$-calculus, but for simplicity, we consider it as a function. We will describe the factorial and count the number of multiplications with $\tick$.
  We write $\Nat[I]$ to denote $\Nat[I,I]$. We use the usual notation $I!$ to represent the factorial function in indices. The process representing factorial and its typing derivation are given in Figure~\ref{f:factorial}. The following type $T$ denotes:
	$$ \ServT{i}{[0,i]}{\Nat[i],\ChT{\Nat[i!]}{\OU{[i,i]}{0}}}{(!\IU{[0,0]}{\infty}. \OU{[0,\infty]}{0})} $$ 
	Denoting a server taking an integer as input, and a return channel on which the factorial of this integer is sent, in $i$ units of time. The usage of this server describes that it can be called anytime. This type is reliable and it would be reliable even if composed with any kind of output $\OU{A_o}{0}$ if we want to call this server.
	Let:
	$$ T' = \ServT{i}{[0,i]}{\Nat[i],\ChT{\Nat[i!]}{\OU{[i,i]}{0}}}{\OU{[0,\infty]}{0})} $$
	$$ S = \ChT{\Nat[(i \minus 1)!]}{(\OU{[i \minus 1,i \minus
            1]}{0} \PAR \IU{[0,0]}{[i \minus 1,i \minus 1]})} =
        S_1 \PAR S_2 $$ where $S_1$ and $S_2$ are obtained by the expected separation of the usage. This type $S$ is reliable under $(i);(i \ge 1)$. Thus, we give the typing described in Figure~\ref{f:factorial}. From the type of $f$, we see on its complexity $[0,i]$ that it does at most a linear number of multiplications. Note that the constraints that appear in a $\mathsf{match}$ are useful since without them, we could not prove $i;(i \le 0) \q i! = 1$ and $i;(i \ge 1) \q i * (i \minus 1)! = i!$. Moreover, polymorphism over indices is necessary in order to find that the recursive call is made on a strictly smaller size $i \minus 1$. \qed
	
	\begin{figure*}
			\footnotesize
		\begin{framed} 
			$$  P := !\In{f}{n,r}. \pifn{n}{\Out{r}{1}}{m}{(\nu r') (\Out{f}{m,r'} \PAR \In{r'}{x}. \tick. \Out{r}{mult(n,x)})} $$		
			\begin{prooftree} 
				\AXC{}
				\UIC{$i;\cdot;n : \Nat[i] \p n : \Nat[i]$}
				\AXC{$i;(i \le 0) \q i! = 1$} 
				\doubleLine
				\UIC{$i;i \le 0; f : T', n : \Nat[i], r : \ChT{\Nat[i!]}{\OU{[i,i]}{0}} \p \Out{r}{1} \C [0,i] $}
				\AXC{$\pi_1$}
				\TIC{$i;\cdot; f : T', n : \Nat[i], r : \ChT{\Nat[i!]}{\OU{[i,i]}{0}} \p \pifn{n}{\Out{r}{1}}{m}{\cdots} \C [0,i] $}
				\UIC{$\void; f : T \p !\In{f}{n,r}. \pifn{n}{\Out{r}{1}}{m}{\cdots} \C [0,0]$}
			\end{prooftree} 
			with the main branch of $\pi_1$ being:
			\begin{prooftree}
				\scriptsize 
				%\AXC{$\cdots \p (m,r') \COL \Isub{(\Nat[i],\ChT{\Nat[i!]}{\OU{[i,i]}{0}})}{i}{i \minus 1}$}
				%\UIC{$i;i \ge 1; m : \Nat[i \minus 1], r' : S_1, f : T_1 \p \Out{f}{m,r'} \C [0,i]$}
				\AXC{$\cdots$}
				\AXC{$(i;i\ge 1) \q i * (i \minus 1)! = i!$} 
				\doubleLine 
				\UIC{$i;i \ge 1;n : \Nat[i] , x : \Nat[(i \minus 1)!] \p mult(n,x) \COL \Nat[i!]$} 
				\UIC{$i;i \ge 1;n : \Nat[i] , x : \Nat[(i \minus 1)!], r : \ChT{\Nat[i!]}{\OU{[0,0]}{0}} \p \Out{r}{mult(n,x)} \C [0,0]$}
				\UIC{$i;i \ge 1;n : \Nat[i] , x : \Nat[(i \minus 1)!], r : \ChT{\Nat[i!]}{\OU{[1,1]}{0}} \p \tick. \Out{r}{mult(n,x)} \C [1,1]$}
				\UIC{$i;i \ge 1;n : \Nat[i] , r' : S_2, r : \ChT{\Nat[i!]}{\OU{[i,i]}{0}} \p \In{r'}{x}. \cdots \C [0,i]$}
				\BIC{$i;i \ge 1;n : \Nat[i], m : \Nat[i \minus 1], r : \ChT{\Nat[i!]}{\OU{[i,i]}{0}}, f : T', r' : S \p \Out{f}{m,r'} \PAR \In{r'}{x}. \cdots \C [0,i]$}
				\UIC{$i;i \ge 1;n : \Nat[i], m : \Nat[i \minus 1], r : \ChT{\Nat[i!]}{\OU{[i,i]}{0}}, f : T' \p (\nu r') \cdots \C [0,i]$}
			\end{prooftree}
		\end{framed} 
		\caption{Representation and Typing of Factorial}
		\label{f:factorial}
	\end{figure*}
	
\end{example} 

Let us now justify the use of this operator $J_c;K$ in order to treat complexity. 

\begin{example}[Deadlock]
	Let us consider the process 
	\(P := (\nu a) (\nu b) (\In{a}{}. \tick. \Out{b}{} \PAR \In{b}{}. \tick. \Out{a}{}) \).
        %%	 In order to understand the constraints we need to verify in order to type this process, we will give a typing with variables for obligation and capacity, and we will look at what values those variables can take. This is described
        \(P\) is typed as shown in Figure~\ref{f:deadlockinf}.  
        %
\begin{comment}
	\begin{figure*}[h]
		\footnotesize 
	\begin{framed} 
	\begin{prooftree}
		\AXC{$\void; a : \ChT{}{\IU{A_o}{I_c}}, b : \ChT{}{\OU{B_o}{J_c}} \p \In{a}{}. \cdots \C K_1 $} 
		\AXC{$\void; a : \ChT{}{\OU{B_o}{J_c}}, b : \ChT{}{\IU{A_o}{I_c}} \p \In{b}{}. \cdots \C K_2 $} 
		\BIC{$\void; a : \ChT{}{(\IU{A_o}{I_c} \PAR \OU{B_o}{J_c})}, b : \ChT{}{(\OU{B_o}{J_c} \PAR \IU{A_o}{I_c})} \p \In{a}{}.\tick.\Out{b}{} \PAR \In{b}{}.\tick.\Out{a}{} \C K_1 \Ilub K_2 $}
		\UIC{$\void;\cdot \p (\nu a) (\nu b) (\In{a}{}. \tick. \Out{b}{} \PAR \In{b}{}. \tick. \Out{a}{}) \C K_1 \Ilub K_2$}
	\end{prooftree}
	\end{framed} 
	\caption{Typing Constraints for Example~\ref{e:deadlock}}
	\label{f:deadlockattempt}
	\end{figure*}
\end{comment}        
%	
	As $a$ and $b$ have exactly the same behaviour,
        let us focus on the typing of \(\In{a}{}.\tick.\Out{b}{}\).
        The derivation for the subprocess \(\tick.\Out{b}{}\) should be clear.
        By assigning the usage to \(\IU{[0,0]}{[\infty,\infty]}\),
        the cost for \(\In{a}{}.\tick.\Out{b}{}\) is calculated by:
        $[\infty,\infty];K = [0,0]$. Thus, we can correctly infer that
        the complexity of the deadlocked process is \(0\).
\qed	
\begin{figure*}[h]
	\begin{framed} 
	\begin{prooftree}
		\footnotesize  
		\AXC{}
		\UIC{$\void; a : \ChT{}{\ZU}, b : \ChT{}{\OU{[0,0]}{0}} \p \Out{b}{} \C [0,0] $}
		\UIC{$\void; a : \ChT{}{\ZU}, b : \ChT{}{\OU{[1,1]}{0}} \p \tick.\Out{b}{} \C [1,1] $} 
		\UIC{$\void; a : \ChT{}{\IU{[0,0]}{[\infty,\infty]}}, b : \ChT{}{\OU{[\infty,\infty]}{0}} \p \In{a}{}.\tick.\Out{b}{} \C [0,0] $}
		%\AXC{}
		%\UIC{$\void; a : \ChT{}{\OU{[0,0]}{0}}, b : \ChT{}{\ZU} \p \Out{a}{} \C [0,0] $}
		%\UIC{$\void; a : \ChT{}{\OU{[1,1]}{0}}, b : \ChT{}{\ZU} \p \tick.\Out{a}{} \C [1,1] $} 
		%\UIC{$\void; a : \ChT{}{\OU{[\infty,\infty]}{0}}, b : \ChT{}{\IU{[0,0]}{[\infty,\infty]}} \p \In{b}{}.\tick.\Out{a}{} \C [0,0] $} 
		\AXC{symmetry for the other branch}
		\BIC{$\void; a : \ChT{}{(\IU{[0,0]}{[\infty,\infty]} \PAR \OU{[\infty,\infty]}{0})}, b : \ChT{}{(\OU{[\infty,\infty]}{0} \PAR \IU{[0,0]}{[0,0]})} \p \In{a}{}.\tick.\Out{b}{} \PAR \In{b}{}.\tick.\Out{a}{} \C [0,0] $}
		\UIC{$\void;\cdot \p (\nu a) (\nu b) (\In{a}{}. \tick. \Out{b}{} \PAR \In{b}{}. \tick. \Out{a}{}) \C [0,0]$}
	\end{prooftree}
	\end{framed} 
	\caption{Typing of Example~\ref{e:deadlock}}
	\label{f:deadlockinf}
\end{figure*}
\label{e:deadlock} 
\end{example}

\begin{example}
	We describe informally an example for which our system can give a complexity, but fails to catch a precise bound. Let us consider the process:
	\small 
	$$ P := \tick. ! \In{a}{n}. \pifn{n}{\pzero}{m}{\Out{a}{m}} \PAR \Out{a}{10} \PAR \tick. \tick. !\In{a}{n}. \pzero $$
	\normalsize
	This process has complexity $2$. However, if we want to give a usage to the server $a$, we must have a usage:
	$$ !\IU{[1,1]}{0}. \OU{[0,0]}{1} \PAR \OU{[0,0]}{[1,2]} \PAR ! \IU{[2,2]}{0} $$
	We took as obligations the number of ticks before the action, and as capacity the minimal number for which we have reliability. So in particular, because of the capacity $1$ in the usage $\OU{[0,0]}{1}$, typing the recursive call $\Out{a}{m}$ increases the complexity by one, and so typing $n$ recursive calls generates a complexity of $n$ in the type system. So, in our setting, the complexity of this process can only be bounded by $10$. Overall, this type system may not behave well when there are more than one replicated input process
	on each server channel, since an imprecision on a capacity for a recursive call leads to an overall imprecision depending on the number of recursive calls. This issue is the only source
    of imprecision we found with respect to the type system of \cite{BaillotGhyselen21}:
    see the conjecture in Section~\ref{relatedwork}. \qed
	\label{e:servers}
	
\end{example}

\section{Related Work}\label{relatedwork}

% Some contributions to the complexity analysis of parallel functional programs by means of type systems appear in \cite{HoffmannShaoParallelPrograms,DBLP:conf/popl/GimenezM16} but the languages studied do not offer the same communication primitives as  the $\pi$-calculus and do not express concurrency.
Some contributions to the complexity analysis of parallel functional programs by types appear in \cite{HoffmannShaoParallelPrograms,DBLP:conf/popl/GimenezM16} but the languages studied  do not express concurrency. 
Alternatively \cite{Albert_SAS2015,AlbertGuaranteeTerminationCostAnalysis,Giachino2015}  address the problem of analysing the time complexity of distributed or concurrent systems. They provide interesting analyses on some instances of systems
but  do not handle dynamic creation of processes and channel name passing as in the   $\pi$-calculus. Moreover, the  flow graph or rely-guarantee reasoning  techniques
employed in \cite{Albert_SAS2015,AlbertGuaranteeTerminationCostAnalysis}  do not seem to offer the same compositionality
as type systems.

There have recently been several studies on
type-based cost analysis for binary or multiparty session
calculi~\cite{Castro-PerezYoshida20,DasHoffmannPfenningTemporalSessionTypes,DasHoffmannPfenningLICS2018}.
It is not clear whether and how those methods can be extended to deal with more general concurrent processes
that can be written in the \(\pi\)-calculus,
where there may be more than one sender/receiver process for each channel.
Among those studies, Das et al.'s work~\cite{DasHoffmannPfenningTemporalSessionTypes} seems technically
closest to ours. Both their cost models and ours are parametric, as they rely on a similar
$\tick$ operation. Moreover, their temporal operators seem to have a strong correspondence with usages and
operations on them. More specifically, the next operator $\bigcirc$~\cite{DasHoffmannPfenningTemporalSessionTypes}
is similar to the usage operator $\Up{[1,1]}$ in our type system,
and $\Box$ and $\lozenge$ roughly correspond to input and output usages with capacity and
obligations \([0,\infty]\). For more precise comparison, we need to extend our type system
with variant and recursive types, to encode their session calculus into the \(\pi\)-calculus,
following~\cite{Kobayashi2003TypeSystemConcurrent,DardhaetalSessionTypesRevisited}. It is left for future work.

% Let us now turn to related works in the setting of $\pi$-calculus or
% process calculi.
%Kobayashi
Kobayashi et al. \cite{Kobayashinewtypesystemdeadlockfree,KobayashiTypeBasedInformationFlowAnalysis,Kobayashi2003TypeSystemConcurrent} also used the notion of usages
to reason about deadlocks, livelocks, and information flow, but he used
a single number for each obligation and capacity (the latter is called a ``capability''
in his work). In particular, a usage type system for time boundedness, related with parallel complexity, was given in \cite{KobayashiTypeSystemLockFree}. However, the definition of parallel complexity and thus the definition of usages and reliability in this work is quite different from ours, as its reduction does not take into account some non-deterministic paths. Moreover, as explained in Example~\ref{e:intervals}, the use of a single number and not intervals induces a loss of precision even on simple examples; we have generalized the number to an interval to improve
the precision of our analysis. 
More recently, %Baillot and Ghyselen
the first two authors proposed in \cite{BaillotGhyselen21} a type system with the same goal of analysing parallel complexity in $\pi$-calculus. %Their
This type system builds on sized types and input/output types instead of usages. Because of that, they cannot manage successive uses of the same channel as in Example~\ref{e:semaphore}, as their names can essentially be used at only one specific time. 
%So, the type system with usage is fairly more precise as there are fewer restrictions on channels. 
In most cases, the time annotation used for channels in their setting corresponds to the sum of the lower bound for obligation and the upper bound for capacity in our setting. We conjecture the following result:
\begin{conjecture}[Comparison with \cite{BaillotGhyselen21}]
	Suppose given a typing $\fifi; \Gamma_{i/o} \p_{i/o} P \C J$ in the input/output sized type system of \cite{BaillotGhyselen21}, such that this process $P$ has a linear use of channels.
	 Then, there exists a reliable context $\Gamma$ such that $\fifi; \Gamma \p P \C [0,J]$. 
\end{conjecture}
 More details and some intuitions are given in the Appendix~\ref{ss:comparison}. So, on a simple use of names our system is strictly more precise if this conjecture is true. However, on other cases, like in Example~\ref{e:servers}, their system is more precise as the loss of precision because of usage does not happen in their setting. On the contrary, our setting has fairly more precision for processes with a non-trivial use of channels, as in Example~\ref{e:semaphore}. 
 %Indeed, their type system cannot express such a sequential use of channels, and so they cannot give a bound.

There have also been studies on implicit computational complexity for process calculi, albeit for less expressive calculi than the pi-calculus \cite{MadetAmadioELL,DiGiamberardinoDalLagoSessionTypePolynomial,DalLagoMartiniSangiorgi16}. Unlike our work,
they consider the work rather than the span, and characterize complexity classes, rather than estimating the precise execution time of a given process. The paper \cite{DemangeonYoshidaSafePiCalculus} by contrast considers the $\pi$-calculus and causal (parallel) complexity, but the goal here is also to delineate a characterization of polynomial complexity.

%Some works have also been carried out in implicit computational complexity to characterize some complexity classes with process calculi \cite{MadetAmadioELL,DiGiamberardinoDalLagoSessionTypePolynomial,DalLagoMartiniSangiorgi16} but for some languages which are not as expressive as the $\pi$-calculus and considering generally the work rather than the span. The paper \cite{DemangeonYoshidaSafePiCalculus}  considers  the $\pi$-calculus and causal (parallel) complexity, but the goal there is more to delineate a characterization of polynomial complexity rather than to give a sharp time bound on the execution of a given process. 

%The paper \cite{DemangeonYoshidaSafePiCalculus} by contrast considers  the $\pi$-calculus and causal (parallel) complexity, but the goal here is more to delineate a %characterization of polynomial complexity rather than to give a sharp bound for a given process. 

% Technically, the types we use are inspired from linear dependent
% types \cite{DalLagoGaboardiLinearDependentTypes}. Those are one of the many
% variants of size types, which were  introduced in
% \cite{HughesParetoSabrySizedTypes}.

\section{Conclusion} 
%\AG{I do not think this section is really necessary. We can talk about implementation in the introduction or maybe somewhere in the paper.}
%\NK{This section could be replaced with
%  a conclusion section, with a brief summary of the contributions and a single sentence
%  to mention an implementation as future work.}
%We plan to investigate how type inference could be automatized or partially automatized for this type system. Building on Kobayashi's work \cite{KobayashiTypeSystemLockFree,Kobayashinewtypesystemdeadlockfree}, we could find procedures to infer usages type system. As we use sized types and time annotations in usages, type inference should reduce to satisfying a set of constraints on indices. We plan to explore whether existing off-the-shelf solvers or new procedures could allow to solve these constraints, on our language or a restricted one.

We presented a type system built on sized types and usages such that a type derivation for a process gives an upper bound on the parallel complexity of this process. The type system relies on intervals in order to give an approximation of the sizes of integers in the process, and an approximation of the time an input or an output needs to synchronize. In comparison to \cite{BaillotGhyselen21}, we showed with examples that our type system can type some concurrent behaviour that was not captured in their type system, and on a certain subset of processes, we conjecture that our new type system is strictly more precise. 

Building on previous work by %Kobayashi
 the third author on type inference for usages~\cite{KobayashiTypeBasedInformationFlowAnalysis,KobayashietAlDeadlockFreeCalculus}, we plan to investigate type inference, with the use of constraint solving procedures for indices.

\newpage 

\bibliography{Bibliography}

\newpage 

\appendix 
\section*{Appendix}

%%<<<<<<< HEAD
%%In the appendix, we describe some proofs elements of lemmas and theorems given in the paper. 
%%\NK{Better to reorganize the appendix, so that the referees can find proofs
%%of the theorems and lemmas  in the main text more easily.}
\section{%%Elements of
Comparison with \cite{BaillotGhyselen21}} 
\label{ss:comparison}

In this section, we give intuitively a description of how to simulate types on \cite{BaillotGhyselen21} in a linear setting with usage. We say that a process has a linear use of channels if it use channel names at most one time for input and at most one time for output. For servers, we suppose that the replicated input is once and for all defined at the beginning of a process, and as free variables it can only use others servers.
In their type system, a channel is given a type $\texttt{Ch}_I(\seq{T})$ where $I$ is an upper bound on the time this channel communicates. It can also be a variant of this type with only input or only output capability. Such a channel would be represented in out type system by a type $\ChT{\seq{\underline{T}}}{(\IU{[I_1,I_1]}{J_c^1} \PAR \OU{I_2,I_2}{J_c^2})}$ where either $J_c^1$ is $0$ and then $I_1 \le I$ , either $J_c^1 = [J_1,J_1]$ and then $I_1 + J_1 \le I$. We have the same thing for $J_c^2$ and $I_2$. To be more precise, the typing in our setting should be a non-deterministic choice (using $+$) over such usages, and the capacity should adapt to the obligation of the dual action in order to be reliable. So, for example if $I_1 \le I_2$, then we would take: 
$\ChT{\seq{\underline{T}}}{(\IU{[I_1,I_1]}{[I_2 \minus I_1,I_2 \minus I_1]} \PAR \OU{I_2,I_2}{0})}$.
Note that this shape of type adapts well to the way time is delayed in their setting. For example, the $\tick$ constructor in their setting make the time advance by $1$, and in our setting, then we would obtain the usage $(\IU{[I_1 + 1,I_1 + 1]}{[I_2 \minus I_1,I_2 \minus I_1]} \PAR \OU{I_2 + 1,I_2 + 1}{0})$ and we still have $I_2 \minus I_1 = (I_2 + 1) \minus (I_1 + 1)$. 

In the same way, in their setting when doing an output (or input), the time is delayed by $I$. Here, with usages, it would be delayed by $J_c$ which is, by definition, a delay of the shape $\Up{[J,J]}$ with $J \le I$. So, we would keep the invariant that our time annotation have the shape of singleton interval with a smaller value than the time annotation in their setting. 

For servers, in the linear setting, their types have the shape: $~_I \ServT{\seq{i}}{J}{\seq{T}}{}$ where $I = 0$ is again a time annotation giving an upper bound on the time the input action of this server is defined, and $J$ is a complexity as in our setting. So, in our setting it would be:
$$ \ServT{\seq{i}}{[0,J]}{\seq{\underline{T}}}{!\IU{[0,0]}{\infty}. !\OU{[0,\infty]}{0} \PAR ! \OU{[0,\infty]}{0}} $$
 Note that this usage is reliable. The main point here is this infinite capacity for input. Please note that because of our input rule for servers, it does not generates an infinite complexity. However, it imposes a delaying $\Up{[0,\infty]} ! \Gamma$ in the context. Because of the shape we gave to types, it means that the context can only have outputs for other servers as free variables, but this was the condition imposed by linearity. Note that in \cite{BaillotGhyselen21}, they have a restriction on the free variables of servers that is in fact the same restriction so it does not harm the comparison to take this restriction on free variables. As an example, the bitonic sort described in \cite{BaillotGhyselen21} could be typed similarly in our setting with this kind of type.  

Finally, choice in usages $U_1 + U_2$ is used to put together the different usages we obtain in the two branches of a pattern matching.

\section{Proofs}
In this section, we prove Theorem~\ref{t:parallelsubjectreduction},
after giving various lemmas. 

\subsection{Properties of Subusage}
%%=======
%%In the appendix, we describe some proofs elements of lemmas and theorems given in the paper. 
%%\section{Properties of Subusage}
%%>>>>>>> 3cc44f03cc9faf544473e258795210672edde09f
\label{ss:subusageprop}

The subusage relation satisfies some properties that are essential for the soundness theorem. First, we have the usual properties of subusage: 

\begin{lemma}
	If $\fifi \q B_o \subseteq A_o$ then $\fifi \p (\Up{A_o}{U}) \SubU (\Up{B_o}{U})$
	\label{l:subusagedelay}
\end{lemma}

This can be proved easily by induction on $U$.

\begin{lemma}[Properties of Subusage]
	For a set of index variables $\varphi$ and a set of constraints $\Phi$ on $\varphi$ we have:
	\begin{enumerate}
		\item If $\fifi \p U \SubU V$ then for any interval $A_o$, we have $\fifi \p \Up{A_o}{U} \SubU \Up{A_o}{V}$. 
		\item If $\fifi \p U \SubU V$ and $\fifi \p V \ured V'$, then there exists $U'$ such that $\fifi \p U \ured^* U'$ and $\fifi \p U' \SubU V'$ (with $\Error \SubU U$ for any usage $U$)
		\item If $\fifi \p U \SubU V$ and $U$ is reliable under $\fifi$ then $V$ is reliable under $\fifi$. 
	\end{enumerate}
	\label{l:subusageprop}
\end{lemma}

The proof of the first lemma is rather easy, but for the second lemma it is a bit cumbersome. Intuitively all the base rules for subusage satisfy this and the main difficulty is to show that the congruence rule and context rule conserve those properties. 

\begin{comment} 

 The first point shows that subtyping is invariant by delaying. The second property means that the subusage relation serves as a simulation relation, and the last one means that the reliability is closed under the subusage relation. In our setting, it means that the subusage relation cannot lead to unsound complexity bounds.
%A proof of the lemma above is given below.
% Some proof elements and other lemmas can be found in the Appendix~\ref{ss:subusageprop}.
 
\end{comment} 

We now give elements of proof for Point~2 and Point~3 of Lemma~\ref{l:subusageprop}. For Point~2, we can see two possible directions, either proceed by induction on $U \SubU V$ and then do a case analysis on $V \ured V'$ on proceed by induction on $V \ured V'$ and then do a case analysis on $U \SubU V$. In both cases, the definition of subusage with the transitivity rule and congruence that can be used everywhere make the proof complicated. So, we chose to first simplify the definition of subusage. 

\begin{definition}[Decomposition of Subusage]
	Let us call $\SubU_{ct}$ the relation defined by the rule of Figure~\ref{f:subusage} without using congruence and transitivity. 
	Then, we define $\SubU_t$ as:
	\begin{prooftree}
		\AXC{$U \congr U' \qquad \fifi \p U' \SubU_{ct} V' \qquad V' \congr V'$}
		\UIC{$\fifi \p U \SubU_t V$}
	\end{prooftree}
\end{definition} 

And we have the following lemma.

\begin{lemma}[Decomposition of Subusage]
	The subusage relation $\SubU$ is equivalent to the reflexive and transitive closure of $\SubU_t$.
	\label{l:decomposesubusage}
\end{lemma}

The proof can be done by induction on $\SubU$, and it relies mainly on the fact that congruence can be used in any context, so we can use it at the end. In the same way, as the subusage relation can also be used in any context, the transitivity can be done at the end. 
So, with this lemma we got rid of the complicated rules by putting them always at the end of a derivation. Moreover, note that it is easy to describe exhaustively subtyping for $\SubU_{t}$.
Let us drop the $\fifi$ notation for the sake of conciseness, and we have: 
\begin{lemma}[Exhaustive Description of Subusage]
	Let us use $* \in \{ \cdot, ! \}$. Then, $*U$ denotes either $U$ or $!U$ according to the value of $*$. If $U \SubU_{t} V$, then one of the following cases hold. 
	\footnotesize 
	\begin{align*}
		&U \congr (U_0 \PAR *U_1) \qquad V \congr U_0 \\
		& U \congr U_0 \PAR *(U_1 + U_2) \qquad V \congr U_0 \PAR *U_i \qquad i \in \set{1;2} \\
		& U \congr U_0 \PAR *\AU{\alpha}{A_o}{J_c}. W \qquad V \congr U_0 \PAR *\AU{\alpha}{A_o}{J_c}. W' \qquad W \SubU_{ct} W' \\
		& U \congr U_0 \PAR *(U_1 + U_2) \qquad V \congr U_0 \PAR *(U_1' + U_2) \qquad U_1 \SubU_{ct} U_1' \\
		& U \congr U_0 \PAR *(U_1 + U_2) \qquad V \congr U_0 \PAR *(U_1 + U_2') \qquad U_2 \SubU_{ct} U_2' \\ 
		&U \congr U_0 \PAR *\AU{\alpha}{A_o}{J_c}. W \qquad V \congr U_0 \PAR *\AU{\alpha}{A_o'}{J_c'}. W \qquad A_o' \subseteq A_o \qquad J_c \le J_c'\\ 
		&U \congr U_0 \PAR *((\AU{\alpha}{A_o}{J_c}. W) \PAR (\Up{A_o + J_c}{U_1})) \qquad V \congr U_0 \PAR *\AU{\alpha}{A_o}{J_c}. (W \PAR U_1) 
	\end{align*} 	
	\label{l:exhaustivesubusage}
\end{lemma} 

This proof is done directly by induction on $U \SubU_{ct} V$. The replication context rules works because we have $!(U_0 \PAR W) \congr !U_0 \PAR !W$ and $!!W \congr !W$. We now go back to Point~2 of Lemma~\ref{l:subusageprop}.

\begin{proof} 
	
	We rely on Lemma~\ref{l:decomposesubusage}. So, we will prove first this intermediate lemma: 
	\begin{lemma}
		If $U \subtype_{t} V$ and $V \red V'$, then there exists $U'$ such that $ U \red^* U'$ and $U' \subtype V'$
		\label{l:subusageprop2int}
	\end{lemma}
	Then, if this lemma is proved, we conclude by transitivity of the propriety.  
	So, we only have to prove 
	We will also use the following lemma
	\begin{lemma}
		If $V \congr U_0 \PAR !U_1$ and $V \ured V'$ then $V' \congr W \PAR !U_1$ with $U_0 \PAR U_1 \ured W$
		\label{l:replicationandreduction}
	\end{lemma}
	Indeed, there are three cases for $V \ured V'$. Either it is only a reduction step in $U_0$ independently of $!U_1$, either it is a reduction step within $U_1$ (note that one copy is always sufficient), either it is a synchronization between one action in $U_0$ and one action in $U_1$. In the first case, the lemma is true because we can arbitrarily  add $U_1$. In the same way, in the second case the lemma is correct because we can just ignore $U_0$. In the third case, the lemma is verified since  we allow this synchronization between $U_0$ and $U_1$. 
	
	We now start the proof. We proceed by induction on $U \subtype_{ct} V$, and we use the exhaustive description given by Lemma~\ref{l:exhaustivesubusage}. We always consider the case $* = !$ as it is the harder of the two cases. Let us give some interesting cases:
	\begin{itemize}
		\item 
		\small
		$$ U \congr U_0 \PAR !\AU{\alpha}{A_o}{J_c}. W \qquad V \congr U_0 \PAR !\AU{\alpha}{A_o}{J_c}. W' \hspace{5mm} W \subtype_{ct} W' $$
		\normalsize 
		The easy case is when $V'$ is obtained by a reduction step in $U_0$. So, we suppose that the reduction step is a synchronization between $U_0$ and $\AU{\alpha}{A_o}{J_c}. W'$. So, we have: 
		\small 
		$$ U_0 \congr U_0' \PAR \AU{\overline{\alpha}}{B_o}{I_c}. W_0 $$
		\normalsize 
		If $V' = \Error$, then we take $U' = \Error$ and concludes this case. Otherwise, we have:
		\small  
		$$V' \congr U_0' \PAR !\AU{\alpha}{A_o}{J_c}. W' \PAR  \Up{(A_o \Ilub B_o)}{(W_0 \PAR W')} $$
		\normalsize 
		So, we take 
		\small 
		$$ U' = U_0' \PAR !\AU{\alpha}{A_o}{J_c}. W \PAR \Up{A_o \Ilub B_o}{(W_0 \PAR W)} $$ 
		\normalsize 
		And, we have indeed:
		$$U \ured  U' \qquad U' \SubU V' $$
		The fact that $U' \subtype V'$ is given by the previous point of Lemma~\ref{l:subusageprop}. 
		\item 
		\small 
		$$U \congr U_0 \PAR !\AU{\alpha}{A_o}{J_c}. W \qquad V \congr U_0 \PAR !\AU{\alpha}{A_o'}{J_c'}. W $$ 
		$$ A_o' \subseteq A_o \qquad J_c \le J_c' $$
		\normalsize 
		Again, the only interesting case is when the synchronization is not only in $U_0$. So, we have: 
		\small 
		$$ U_0 \congr U_0' \PAR !\AU{\overline{\alpha}}{B_o}{I_c}. W_0 $$
		\normalsize 
		If we have $A_o \subseteq B_o \IPlus I_c$ and $B_o \subseteq A_o \IPlus J_c$
		then 
		\small 
		$$ V' \congr U_0' \PAR !\AU{\alpha}{A_o'}{J_c'}. W \PAR \Up{A_o' \Ilub B_o}{(W_0 \PAR W)} $$ 
		\normalsize 
		because $A_o' \subseteq A_o$ and $J_c \le J_c'$ so we have $A_o' \subseteq B_o \IPlus I_c$ and $B_o \subseteq A_o \IPlus J_c'$. Thus, we take 
		\small 
		$$ U' = U_0' \PAR !\AU{\alpha}{A_o}{J_c}. W \PAR \Up{A_o \Ilub B_o}{(W_0 \PAR W)} $$
		\normalsize  
		We have indeed $U \ured U'$ and $U' \SubU V'$ by Lemma~\ref{l:subusagedelay}. Otherwise, we obtain an error for $U'$ and so it indeed is a subusage of $V'$. 
		\item 
		\small 
		$$U \congr U_0 \PAR !((\AU{\alpha}{A_o}{J_c}. W) \PAR (\Up{A_o + J_c}{U_1})) \qquad  V \congr U_0 \PAR !\AU{\alpha}{A_o}{J_c}. (W \PAR U_1) $$
		\normalsize 
		Again, if the reduction step $V \red V'$ is a synchronization between subprocesses in $U_0$, it is simple. So let us suppose that:
		\small 
		$$ U_0 \congr U_0' \PAR \AU{\overline{\alpha}}{B_o}{I_c}. W_0 $$
		\normalsize 
		If we have $B_o \subseteq A_o \IPlus J_c$ and $A_o \subseteq B_o \IPlus I_c$, then 
		\small 
		$$ V' \congr U_0' \PAR !\AU{\alpha}{A_o}{J_c}. (W \PAR U_1) \PAR (\Up{A_o \Ilub B_o}{(W \PAR U_1 \PAR W_0)}) $$
		\normalsize  
		We pose $U'$ equal to:
		\small  
		$$ U_0' \PAR !((\AU{\alpha}{A_o}{J}. W) \PAR (\Up{A_o + J_c}{U_1})) \PAR (\Up{A_o \Ilub B_o}{(W \PAR W_0)}) $$
		\normalsize  
		We have indeed $U \red U'$ and we have $U' \subtype V'$ because $A_o \Ilub B_o \subseteq A_o + J_c$.  Indeed, 
		\small 
		$$ \LeA(A_o + J_c) \le \max(\LeA(A_o),\LeA(B_o)) $$
		\normalsize 
		As either $J_c = J$ and so $\LeA(A_o = J_c) = \LeA(A_o)$, either $J_c = [I,J]$ and 
		\small 
		$$ \LeA(A_o) + I \le \RiA(A_o) + I \le \LeA(B_o) $$
		\normalsize 
		since $B_o \subseteq A_o \IPlus J_c$. Moreover, we have 
		\small 
		$$ \max(\RiA(A_o),\RiA(B_o)) \le \RiA(A_o + J_c) $$
		\normalsize 
		again because $B_o \subseteq A_o \IPlus J_c$. 
	\end{itemize}
	Thus, we have indeed Lemma~\ref{l:subusageprop2int}, and we deduce the second point of Lemma~\ref{l:subusageprop}. 
	
	Finally, the third point is a direct consequence of the third point. Indeed, suppose that $U$ is reliable. So, for any reduction from $U$, it does not lead to an error. Let us take a reduction from $V$. By the third point, it gives us a reduction from $U$ where some steps are a subtype of the steps in $V$. So, as an error cannot happen in the steps from $U$, and the only usage $U$ such that $U \subtype \Error$ is $\Error$, we know that the reduction from $V$ does not lead to an error. Thus, $V$ is reliable.
\end{proof}

More specific to this type system, we also have the following lemma, stating that delaying does not modify the behavior of a type.

\begin{lemma}[Invariance by Delaying]
	For any interval $A_o$: 
	\begin{enumerate}
		\item If $ \fifi \p \Up{A_o} U \ured V'$ then, there exists $V$ such that $\Up{A_o} V = V'$ and $\fifi \p U \ured V$. (with $\Error = \Up{A_o} \Error$)
		\item If $\fifi \p U \ured V$ then $\fifi \p \Up{A_o} U \ured \Up{A_o} V$.
		\item $U$ is reliable under $\fifi$ if and only if $\Up{A_o} U$ is reliable under $\fifi$.
	\end{enumerate}
	\label{l:delayinginvariance}
\end{lemma}

In our setting, this lemma shows among other things that the $\tick$ constructor, or more generally the annotation $n : P$, does not break reliability. Those properties can easily be proved with simple inductions. 

\subsection{Intermediate Lemmas} 
\label{ss:intermediate} 
We give some intermediate lemmas for the soundness theorem
\begin{lemma}[Weakening]
Let $\varphi,\varphi'$ be disjoint set of index variables, $\Phi$ be a set of constraints on $\varphi$, $\Phi'$ be a set of constraints on $(\varphi,\varphi')$, $\Gamma$ and $\Gamma'$ be contexts on disjoint set of variables. 
\begin{enumerate}	
	\item If $\fifi ; \Gamma \p e \COL T$ then $(\varphi,\varphi') ; (\Phi,\Phi') ; \Gamma,\Gamma' \p e \COL T$.
	\item If $\fifi ; \Gamma \p P \C K$ then $(\varphi,\varphi') ; (\Phi,\Phi') ; \Gamma,\Gamma' \p P \C K$.  
\end{enumerate}
\label{l:weakening}
\end{lemma}

We also show that we can remove some useless hypothesis. 

\begin{lemma}[Strengthening]
Let $\varphi$ be a set of index variables, $\Phi$ be a set of constraints on $\varphi$, and $C$ a constraint on $\varphi$ such that $\fifi \q C$. 
\begin{enumerate}
	\item If $\varphi;(\Phi,C);\Gamma,\Gamma' \p e \COL T$ and the variables in $\Gamma'$ are not free in $e$, then $\fifi;\Gamma \p e \COL T$. 
	\item If $\varphi;(\Phi,C);\Gamma,\Gamma' \p P \C K$ and the variables in $\Gamma'$ are not free in $P$, then $\fifi;\Gamma \p P \C K$.  
\end{enumerate}
\label{l:strengthening}
\end{lemma}

Those two lemmas are proved easily by successive induction on the definitions in this paper. Then, we also have a lemma expressing that index variables can indeed be replaced by any index.  

\begin{lemma}[Index Substitution]
Let $\varphi$ be a set of index variables and $i \notin \varphi$. Let $\INN{J}$ be an index with free variables in $\varphi$. Then, 
\begin{enumerate}
	\item If $(\varphi,i);\Phi;\Gamma \p e \COL T$ then  $\varphi;\Isub{\Phi}{i}{\INN{J}};\Isub{\Gamma}{i}{\INN{J}} \p e \COL \Isub{T}{i}{\INN{J}}$. 
	\item If $(\varphi,i);\Phi;\Gamma \p P \C K$ then $\varphi;\Isub{\Phi}{i}{\INN{J}};\Isub{\Gamma}{i}{\INN{J}} \p P \C \Isub{K}{i}{\INN{J}}$.  
\end{enumerate}
\label{l:indexsub}
\end{lemma}

\subsection{Substitution Lemma} \label{ss:substitution}

We now present the variable substitution lemmas. In the setting of usages, this lemma is a bit more complex than usual. Indeed, we have a separation of contexts with the parallel composition, and we have to rely on subusage, especially the rule $\fifi \p (\alpha^{A_o}_{J} . U) \PAR (\Up{A_o + J_c}V) \SubU \alpha^{A_o}_{J_c}. (U \PAR V)$ as expressed in the Example~\ref{e:substitution} above.
We put some emphasis on the following notation: when we write $\Gamma, \var : T$ as a context in typing, it means that $\var$ does not appear in $\Gamma$.  

\begin{lemma}[Substitution]
	Let $\Gamma$ and $\Delta$ be contexts such that $\Gamma \PAR \Delta$ is defined. Then we have:
	\begin{enumerate}
		\item If $\fifi; \Gamma, \var \COL T \p e' \COL T'$ and $\Delta \p e \COL T$ then 
		$\fifi; \Gamma \PAR \Delta \p \psub{e'}{\var}{e} \COL T'$ 
		\item If $\fifi; \Gamma, \var \COL T \p P \C K$ and $\Delta \p e \COL T$ then 
		$\fifi; \Gamma \PAR \Delta \p \psub{P}{\var}{e} \C K$ 
	\end{enumerate}
	\label{l:subsitution}
\end{lemma}

The first point is straightforward. It uses the fact that we have the relation
$\fifi \p U \SubU \ZU$ for any usage $U$, and so we can use $\fifi \p \Gamma \PAR \Delta \SubU \Gamma$ in order to weaken $\Delta$ (similarly for $\Gamma$) if needed. 
The second point is more interesting. The easy case is when $T$ is $\Nat[I,J]$ for some $[I,J]$. Then, we take a $\Delta$ that only uses the zero usage, and so $\Gamma \PAR \Delta = \Gamma$ and everything becomes simpler. The more interesting cases are:

\begin{lemma}[Difficult Cases of Substitution] 
	We have:
	\begin{itemize}
		\item If $\fifi; \Gamma, b \COL \ChT{\seq{S}}{W_0}  , c \COL \ChT{\seq{S}}{W_1} \p P \C K$ then $\fifi; \Gamma, b \COL \ChT{\seq{S}}{(W_0 \PAR W_1)} \p \psub{P}{c}{b} \C K$
		\item If $\fifi; \Gamma, b \COL\ServT{\seq{i}}{K}{\seq{S}}{W_0}  , c \COL \ServT{\seq{i}}{K}{\seq{S}}{W_1} \p P \C K$ then \\ $\fifi; \Gamma, b \COL \ServT{\seq{i}}{K}{\seq{S}}{(W_0 \PAR W_1)} \p \psub{P}{c}{b} \C K$	
	\end{itemize} 
\end{lemma}

With a careful induction, we can indeed prove this lemma, relying on the subusage relation.

\begin{enumerate}
	\item 
	\begin{itemize}
		\item \textbf{Case} of input, with $a \ne b$ and $a \ne c$.
		\begin{center}
			\small 
			\AXC{$\fifi;\Gamma, b \COL \ChT{\seq{S}}{W_0}  , c \COL \ChT{\seq{S}}{W_1},  a \COL \ChT{\seq{T}}{U}, \seq{\var} \COL \seq{T} \p P \C K$}
			\UIC{$\fifi;\Up{J_c}{\Gamma}, b \COL \ChT{\seq{S}}{(\Up{J_c}W_0)}  , c \COL \ChT{\seq{S}}{(\Up{J_c}W_1)}, a \COL \ChT{\seq{T}}{\IU{[0,0]}{J_c}.U} \p \In{a}{\seq{\var}}.P \C J_c;K$}
			\DP 
		\end{center}
		
		By induction hypothesis, we obtain $\fifi;\Gamma, b \COL \ChT{\seq{S}}{(W_0 \PAR W_1)},  a \COL \ChT{\seq{T}}{U}, \seq{\var} \COL \seq{T} \p \psub{P}{c}{b} \C K$. 
		
		We then give the following proof: 
		
		\begin{center}
			\small 
			\AXC{$\fifi;\Gamma, b \COL \ChT{\seq{S}}{(W_0 \PAR W_1)},  a \COL \ChT{\seq{T}}{U}, \seq{\var} \COL \seq{T} \p \psub{P}{c}{b} \C K$}
			\UIC{$\fifi;\Up{J_c}{\Gamma}, b \COL \ChT{\seq{S}}{(\Up{J_c}(W_0 \PAR W_1))}, a \COL \ChT{\seq{T}}{\IU{[0,0]}{J_c}.U} \p \In{a}{\seq{\var}}.P \C J_c;K$}
			\DP 
		\end{center}
	
		This case is similar to all other cases when $b$ and $c$ does not interfere with the typing rule. Thus, we know only show the cases when they interfere. 	
		\item \textbf{Case} of input, with $a = b$. 
		
		\begin{center}
			\small 
			\AXC{$\fifi;\Gamma, b \COL \ChT{\seq{T}}{U}, c \COL \ChT{\seq{T}}{W_1}, \seq{\var} \COL \seq{T} \p P \C K$}
			\UIC{$\fifi;\Up{J_c}{\Gamma}, b \COL \ChT{\seq{T}}{\IU{[0,0]}{J_c}.U}, c \COL \ChT{\seq{T}}{(\Up{J_c} W_1)} \p \In{b}{\seq{\var}}.P \C J_c;K$}
			\DP 
		\end{center} 	
		
		By induction hypothesis, we obtain $\fifi;\Gamma, b \COL \ChT{\seq{T}}{(U \PAR W_1)}, \seq{\var} \COL \seq{T} \p \psub{P}{c}{b} \C K$.
		So, we give the typing:
		
		\begin{center}
			\small 
			\AXC{$\fifi;\Gamma, b \COL \ChT{\seq{T}}{(U \PAR W_1)}, \seq{\var} \COL \seq{T} \p \psub{P}{c}{b} \C K$}
			\UIC{$\fifi;\Up{J_c}{\Gamma}, b \COL \ChT{\seq{T}}{\IU{[0,0]}{J_c}.(U \PAR W_1)} \p \In{b}{\seq{\var}}. \psub{P}{c}{b} \C J_c;K$}
			\UIC{$\fifi;\Up{J_c}{\Gamma}, b \COL \ChT{\seq{T}}{\IU{[0,0]}{J_c}.U \PAR (\Up{J_c} W_1)} \p \In{b}{\seq{\var}}. \psub{P}{c}{b} \C J_c;K$}
			\DP 
		\end{center} 
		
		Indeed, the last rule represents subtyping. This concludes this case. 
		\item \textbf{Case} of input, with $a = c$. 
		
		\begin{center}
			\small 
			\AXC{$\fifi;\Gamma, b \COL \ChT{\seq{T}}{W_0}, c \COL \ChT{\seq{T}}{U}, \seq{\var} \COL \seq{T} \p P \C K$}
			\UIC{$\fifi;\Up{J_c}{\Gamma}, b \COL \ChT{\seq{T}}{(\Up{J_c} W_0)}, c \COL \ChT{\seq{T}}{\IU{[0,0]}{J_c}.U} \p \In{c}{\seq{\var}}.P \C J_c;K$}
			\DP 
		\end{center}
		
		By induction hypothesis, we obtain $\fifi;\Gamma, b \COL \ChT{\seq{T}}{(W_0 \PAR U)} ,\seq{\var} \COL \seq{T} \p \psub{P}{c}{b} \C K$. So, we give the typing:
		
		\begin{center}
			\small 
			\AXC{$\fifi;\Gamma, b \COL \ChT{\seq{T}}{(W_0 \PAR U)}, \seq{\var} \COL \seq{T} \p \psub{P}{c}{b} \C K$}
			\UIC{$\fifi;\Up{J_c}{\Gamma}, b \COL \ChT{\seq{T}}{\IU{[0,0]}{J_c}.(W_0 \PAR U)} \p \In{b}{\seq{\var}}. \psub{P}{c}{b} \C J_c;K$}
			\UIC{$\fifi;\Up{J_c}{\Gamma}, b \COL \ChT{\seq{T}}{\IU{[0,0]}{J_c}.U \PAR (\Up{J_c} W_0)} \p \psub{(\In{c}{\seq{\var}}.P)}{c}{b} \C J_c;K$}
			\DP 
		\end{center} 
		
		\item \textbf{Case} of output, with $a = b$. 
		
		\begin{center}
			\small 
			\AXC{$\fifi;\Gamma', b \COL \ChT{\seq{T}}{V}, c \COL \ChT{\seq{T}}{W_1'}  \p \seq{e} \COL \seq{T} \qquad \fifi;\Gamma, b \COL \ChT{\seq{T}}{U}, c \COL \ChT{\seq{T}}{W_1} \p P \C K$}
			\UIC{$\fifi;\Up{J_c}{(\Gamma \PAR \Gamma')}, b \COL \ChT{\seq{T}}{\OU{[0,0]}{J_c}.(V \PAR U)}, c \COL \ChT{\seq{T}}{\Up{J_c} (W_1' \PAR W_1)} \p \Out{b}{\seq{e}}.P \C J_c;K$}
			\DP 
		\end{center}
		
		By point 1 of Lemma~\ref{l:subsitution} and induction hypothesis, we obtain $\fifi;\Gamma', b \COL \ChT{\seq{T}}{(V \PAR W_1')} \p \seq{e} \COL \seq{T}$ and $\fifi;\Gamma, b \COL \ChT{\seq{T}}{(U \PAR W_1)} \p P \C K$. 
		Thus, we have: 
		
		\begin{center}
			\small 
			\AXC{$\fifi;\Gamma', b \COL \ChT{\seq{T}}{(V \PAR W_1')} \p \seq{e} \COL \seq{T} \qquad \fifi;\Gamma, b \COL \ChT{\seq{T}}{(U \PAR W_1)} \p P \C K$}
			\UIC{$\fifi;\Up{J_c}{(\Gamma \PAR \Gamma')}, b \COL \ChT{\seq{T}}{\OU{[0,0]}{J_c}.(V \PAR U \PAR W_1 \PAR W_1')} \p \Out{b}{\seq{e}}.P \C J_c;K$}
			\UIC{$\fifi;\Up{J_c}{(\Gamma \PAR \Gamma')}, b \COL \ChT{\seq{T}}{\OU{[0,0]}{J_c}.(V \PAR U) \PAR \Up{J_c}(W_1 \PAR W_1')} \p \Out{b}{\seq{e}}.P \C J_c;K$}
			\DP 
		\end{center}
		
		Again, the last rule is obtained by subtyping. We have a similar proof for the case $a = c$.
	\end{itemize}
	
	\item We know work on the case of servers. The notations are a bit cumbersome but the proofs are similar to the one for channels. The only point that need some details is for server input as there is the replication in usages that appear. 
	
	\begin{itemize}
		\item \textbf{Case} of input, with $a \ne b$ and $a \ne c$. 
		
		\begin{center}
			\scriptsize 
			\AXC{$(\varphi,\seq{i});\Phi;\Gamma, a \COL \ServT{\seq{i}}{K}{\seq{T}}{U}, b \COL \ServT{\seq{j}}{L}{\seq{S}}{W_0}  , c \COL \ServT{\seq{j}}{L}{\seq{S}}{W_1},  \seq{\var} \COL \seq{T} \p P \C K$}
			\UIC{$\fifi;\Up{J_c}!\Gamma, a \COL \ServT{\seq{i}}{K}{\seq{T}}{!\IU{[0,0]}{J_c}.U}, b \COL \ServT{\seq{j}}{L}{\seq{S}}{(\Up{J_c} !W_0)}  , c \COL \ServT{\seq{j}}{L}{\seq{S}}{(\Up{J_c} !W_1)} \p !\In{a}{\seq{\var}}.P \C [0,0]$}
			\DP
		\end{center}
		
		By induction hypothesis, we have $(\varphi,\seq{i});\Phi;\Gamma, a \COL \ServT{\seq{i}}{K}{\seq{T}}{U}, b \COL \ServT{\seq{j}}{L}{\seq{S}}{(W_0 \PAR W_1)}   , \seq{\var} \COL \seq{T} \p \psub{P}{c}{b} \C K$. So, we have the proof
		
		\begin{center}
			\footnotesize 
			\AXC{$(\varphi,\seq{i});\Phi;\Gamma, a \COL \ServT{\seq{i}}{L}{\seq{T}}{U}, b \COL \ServT{\seq{j}}{L}{\seq{S}}{(W_0 \PAR W_1)}, \seq{\var} \COL \seq{T} \p \psub{P}{c}{b} \C K$}
			\UIC{$\fifi;\Up{J_c}!\Gamma, a \COL \ServT{\seq{i}}{K}{\seq{T}}{!\IU{[0,0]}{J_c}.U}, b \COL \ServT{\seq{j}}{L}{\seq{S}}{(\Up{J_c} !(W_0 \PAR W_1))} \p !\In{a}{\seq{\var}}.\psub{P}{c}{b} \C [0,0]$}
			\DP
		\end{center} 
		
		\item \textbf{Case} of input, with $a = b$. 
		
		\begin{center}
			\small 
			\AXC{$(\varphi,\seq{i});\Phi;\Gamma, b \COL \ServT{\seq{i}}{K}{\seq{T}}{U}, c \COL \ServT{\seq{i}}{K}{\seq{T}}{W_1},  \seq{\var} \COL \seq{T} \p P \C K$}
			\UIC{$\fifi;\Up{J_c}!\Gamma, b \COL \ServT{\seq{i}}{K}{\seq{T}}{!\IU{[0,0]}{J_c}.U}, c \COL \ServT{\seq{i}}{K}{\seq{T}}{(\Up{J_c} !W_1)} \p !\In{a}{\seq{\var}}.P \C [0,0]$}
			\DP
		\end{center}
		
		By induction hypothesis, we obtain $(\varphi,\seq{i});\Phi;\Gamma, b \COL \ServT{\seq{i}}{K}{\seq{T}}{(U \PAR W_1)}, \seq{\var} \COL \seq{T} \p P \C K$
		
		\begin{center}
			\small 
			\AXC{$(\varphi,\seq{i});\Phi;\Gamma, b \COL \ServT{\seq{i}}{K}{\seq{T}}{(U \PAR W_1)}, \seq{\var} \COL \seq{T} \p \psub{P}{c}{b} \C K$}
			\UIC{$\fifi;\Up{J_c}!\Gamma, b \COL \ServT{\seq{i}}{K}{\seq{T}}{!\IU{[0,0]}{J_c}.(U \PAR W_1)} \p !\In{a}{\seq{\var}}.\psub{P}{c}{b} \C [0,0]$}
			\UIC{$\fifi;\Up{J_c}!\Gamma, b \COL \ServT{\seq{i}}{K}{\seq{T}}{(!\IU{[0,0]}{J_c}.U \PAR \Up{J_c} !W_1)} \p !\In{a}{\seq{\var}}.\psub{P}{c}{b} \C [0,0]$}
			\DP
		
	\end{center}
		This last derivation is obtained by subtyping. Indeed, by definition we have $\Up{J_c} !W_1 = ! \Up{J_c} W_1$. Then, 
		$$ !\IU{[0,0]}{J_c}.U \PAR ! \Up{J_c} W_1 \congr ! (\IU{[0,0]}{J_c}.U \PAR \Up{J_c} W_1) \SubU !\IU{[0,0]}{J_c}.(U \PAR W_1) $$
		\item The case $a = c$ is similar to the previous one. 
	\end{itemize}
\end{enumerate}

This concludes the proof.

\subsection{Congruence Equivalence}
\label{ss:congruence}

In order to prove the soundness theorem, we need first a lemma saying that the congruence relation behaves well with typing. 

\begin{lemma}[Congruence and Typing]
	Let $P$ and $Q$ be annotated processes such that $P \congr Q$. Then, $\fifi; \Gamma \p P \C K$ if and only if $\fifi; \Gamma \p Q \C K$.
	\label{l:parallelcongrtype}
\end{lemma}

The proof is an induction on $P \congr Q$, relying on all the previous lemma, and especially the lemmas on delaying for congruence rules specific to annotated processes. 

\begin{proof}  
	Note that for a process $P$, the typing system is not syntax-directed because of the subtyping rule. However, by reflexivity and transitivity of subtyping, we can always assume that a proof has exactly \emph{one} subtyping rule before any syntax-directed rule. Moreover, notice that in those kinds of proof, the top-level rule of subtyping can be ignored. Indeed, we can always simulate exactly the same subtyping rule for both $P$ and $Q$
	We first show this propriety for base case of congruence. The reflexivity is trivial then we have those interesting cases:
	
	\begin{itemize}
		
		\item \textbf{Case} $(\nu a) P \PAR Q \congr (\nu a) (P \PAR Q)$ with $a$ not free in $Q$. Suppose $\fifi;\Gamma \PAR \Delta \p {(\nu a) P \PAR Q} \C K$. Then the proof has the shape:
		
		\begin{prooftree}
			\footnotesize
			\AXC{$\pi$}
			\UIC{$\fifi;\Gamma', a : T \p {P} \C K_1'$}
			\AXC{$T$ reliable} 
			\BIC{$\fifi;\Gamma' \p {(\nu a) P} \C K_1'$}
			\AXC{$\fifi \p \Gamma \subtype \Gamma' ; K_1' \subseteq K_1$}
			\BIC{$\fifi;\Gamma \p {(\nu a)P} \C K_1$}
			\AXC{$\pi'$}
			\UIC{$\fifi;\Delta \p {Q} \C K_2$}
			\BIC{$\fifi; \Gamma \PAR \Delta \p {(\nu a)P \PAR Q} \C K_1 \Ilub K_2$}
		\end{prooftree}
		
		By weakening (Lemma~\ref{l:weakening}), we obtain a proof $\pi'_w$ of $\fifi ;\Delta, a : (T / \ZU) \p Q \C K_2$. Thus, we have the following derivation:
		
		\begin{prooftree}
			\scriptsize 
			\AXC{$\pi$}
			\UIC{$\fifi; \Gamma', a : T \p {P} \C K_1'$}
			\AXC{$\fifi \p \Gamma \subtype \Gamma' ; K_1' \subseteq K_1$}
			\BIC{$\fifi; \Gamma, a : T \p {P} \C K_1$}
			\AXC{$\pi'_w$}
			\UIC{$\fifi; \Delta, a : (T / \ZU) \p {Q} \C K_2$}
			\BIC{$\fifi;\Gamma \PAR \Delta , a : T \p {P \PAR Q} \C K_1 \Ilub K_2$}
			\AXC{$T$ reliable} 
			\BIC{$\fifi; \Gamma \PAR \Delta \p {(\nu a)(P \PAR Q)} \C K_1 \Ilub K_2$}
		\end{prooftree}
		
		For the converse, suppose $\fifi; \Gamma \p {(\nu a)(P \PAR Q)} \C K$. Then the proof has the shape:
		
		\begin{prooftree}
			\small 	
			\AXC{$\pi$}
			\UIC{$\fifi; \Gamma_P, a : T_P \p {P} \C K_1$}
			\AXC{$\pi'$}
			\UIC{$\fifi;\Gamma_Q, a : T_Q \p {Q} \C K_2$}
			\BIC{$\fifi; \Gamma_P \PAR \Gamma_Q, a : T_P \PAR T_Q \p {P \PAR Q} \C K_1 \Ilub K_2$}
			\AXC{$(1)$}
			\BIC{$\fifi;\Gamma, a : T \p {P \PAR Q} \C K$}
			\UIC{$\fifi;\Gamma \p {(\nu a)(P \PAR Q)} \C K$}
		\end{prooftree}
		
		with $T$ reliable, and where $(1)$ is:
		$$\fifi \p \Gamma \subtype \Gamma_P \PAR \Gamma_Q ; T \subtype T_P \PAR T_Q ; K_1 \Ilub K_2 \subseteq K$$ 
		Since $a$ is not free in $Q$, by Lemma~\ref{l:strengthening}, from $\pi'$ we obtain a proof $\pi'_s$ of $\fifi; \Gamma_Q \p {Q} \C K_2$. We then derive the following typing:
		
		\begin{prooftree}
			\scriptsize 
			\AXC{$\pi$}
			\UIC{$\fifi;\Gamma_P, a : T_P \p P \C K_1$}
			\AXC{$\fifi \p T \subtype T_P \PAR T_Q \subtype T_P$}
			\BIC{$\fifi;\Gamma_P, a : T \p {P} \C K_1$}
			\AXC{$T$ reliable} 
			\BIC{$\fifi;\Gamma_P \p {(\nu a)P} \C K_1$}
			\AXC{$\pi'_s$}
			\UIC{$\fifi;\Gamma_Q \p {Q} \C K_2$}
			\BIC{$\fifi;\Gamma_P \PAR \Gamma_Q \p {(\nu a)P \PAR Q} \C K_1 \Ilub K_2 \qquad \fifi \p \Gamma \subtype \Gamma_P \PAR \Gamma_Q ; K_1 \Ilub K_2 \subseteq K$}
			\UIC{$\fifi ; \Gamma \p {(\nu a)P \PAR Q} \C K$}
		\end{prooftree}		
		
		\item \textbf{Case} $m : (P \PAR Q) \congr m : P \PAR m : Q$. Suppose $\fifi; \Up{[m,m]} \Gamma \p  {m : (P \PAR Q)}\C K + [m,m] $. Then we have: 
		
		\begin{prooftree}
			\footnotesize
			\AXC{$\pi_P$}
			\UIC{$\fifi; \Gamma_P \p {P} \C K_1 $}
			\AXC{$\pi_Q$}  
			\UIC{$\fifi; \Gamma_Q \p {Q}\C K_2 $}
			\BIC{$\fifi;\Gamma_P \PAR \Gamma_Q \p {(P \PAR Q)}\C K_1 \Ilub K_2 $} 
			\AXC{$\fifi \p  \Gamma \subtype  \Gamma_P \PAR \Gamma_Q ; K_1 \Ilub K_2 \subseteq K$} 
			\BIC{$\fifi;\Gamma \p {(P \PAR Q)} \C K $} 
			\UIC{$\fifi; \Up{[m,m]} \Gamma \p  {m : (P \PAR Q)}\C K + [m,m]$}
		\end{prooftree}
		
		By Lemma~\ref{l:subusageprop}, from $\fifi \p \Gamma \subtype \Gamma_P \PAR \Gamma_Q$ we obtain $\fifi \p \Up{[m,m]} \Gamma \subtype (\Up{[m,m]} \Gamma_P) \PAR (\Up{[m,m]} \Gamma_Q)$.
		So, we give the following derivation: 
		
		\begin{prooftree}
			\footnotesize
			\AXC{$\pi_P$}
			\UIC{$\fifi;\Gamma_P \p P \C K_1 $}
			\UIC{$\fifi; \Up{[m,m]} \Gamma_P \p m: P \C K_1 + [m,m] $}
			\AXC{$\pi_Q$}
			\UIC{$\fifi;\Gamma_Q \p Q \C K_2 $}
			\UIC{$\fifi; \Up{[m,m]} \Gamma_Q \p m : Q  \C K_2 + [m,m]$}
			\BIC{$\fifi;(\Up{[m,m]} \Gamma_P) \PAR (\Up{[m,m]} \Gamma_Q) \p m : P \PAR m: Q \C (K_1 \Ilub K_2) + [m,m]$}
			\AXC{(1)} 
			\BIC{$\fifi; \Up{[m,m]} \Gamma \p  {m : P \PAR m : Q} \C (K_1 \Ilub K_2) + [m,m] \qquad \fifi \q K_1 \Ilub K_2 \subseteq K $}
			\UIC{$\fifi; \Up{[m,m]} \Gamma \p  {m : P \PAR m : Q} \C K + [m,m] $}
		\end{prooftree}
	
		where $(1)$ is
		$$ \fifi \p  \Up{[m,m]} \Gamma \subtype  (\Up{[m,m]} \Gamma_P) \PAR (\Up{[m,m]} \Gamma_Q)$$ 
		
		Now, suppose we have a typing $\fifi; \Gamma_P \PAR \Gamma_Q \p  {m : P \PAR m : Q} \C K_1 \Ilub K_2 $. The typing has the shape: 
		
		\begin{prooftree}
			\footnotesize
			\AXC{$\pi_P$}
			\UIC{$\fifi;\Delta_P \p {P} \C {K_P}$}
			\UIC{$\fifi;\Up{[m,m]} \Delta_P \p {m : P } \C K_P + [m,m]$} 
			\doubleLine 
			\UIC{$\fifi; \Gamma_P \p {m: P} \C K_1 $}
			\AXC{$\pi_Q$}
			\UIC{$\fifi;\Delta_Q \p {Q} \C {K_Q}$}
			\UIC{$\fifi;\Up{[m,m]} \Delta_Q \p {m : Q } \C K_Q + [m,m]$} 
			\doubleLine 
			\UIC{$\fifi; \Gamma_Q \p {m: Q} \C K_2 $}
			\BIC{$\fifi; \Gamma_P \PAR \Gamma_Q \p  {m : P \PAR m : Q} \C K_1 \Ilub K_2 $}
		\end{prooftree}
		
		with 
		$$\fifi \p \Gamma_P \subtype \Up{[m,m]} \Delta_P \qquad \fifi \p \Gamma_Q \subtype \Up{[m,m]} \Delta_Q $$
		$$\fifi \q K_P + [m,m] \subseteq K_1 \qquad \fifi \q  K_Q + [m,m] \subseteq K_2  $$
		
		So, we derive:
		
		\begin{prooftree}
			\small 
			\AXC{$\pi_P$}
			\UIC{$\fifi;\Delta_P \p {P} \C {K_P}$}
			\AXC{$\pi_Q$}
			\UIC{$\fifi;\Delta_Q \p {Q} \C {K_Q}$}
			\BIC{$\fifi; \Delta_P \PAR \Delta_Q \p  (P \PAR Q) \C K_P \Ilub K_Q$} 
			\UIC{$\fifi; \Up{[m,m]}(\Delta_P \PAR \Delta_Q) \p  {m : (P \PAR Q)} \C (K_P \Ilub K_Q) + [m,m] \qquad (1)$}
			\UIC{$\fifi; \Gamma_P \PAR \Gamma_Q \p  {m : (P \PAR Q)} \C K $}
		\end{prooftree}
		
		where $(1)$ is 
		$$\fifi \p \Gamma_P \PAR \Gamma_Q \subtype \Up{[m,m]}(\Delta_P \PAR \Delta_Q); (K_P \Ilub K_Q) + [m,m] \subseteq K_1 \Ilub K_2$$

		This concludes this case.
		
		\item \textbf{Case} $m : (\nu a) P \congr (\nu a)(m : P)$.
		
		Suppose $\fifi; \Up{[m,m]} \Gamma \p m : (\nu a) P \C K + [m,m]$. Then, the typing has the shape:
		
		\begin{prooftree}
			\footnotesize
			\AXC{$\pi$}
			\UIC{$\fifi; \Gamma', a : T \p P \C K'$}
			\AXC{$T$ reliable}
			\BIC{$\fifi; \Gamma' \p (\nu a) P \C K'$}
			\AXC{$\fifi \p \Gamma \subtype \Gamma' ; K' \subseteq K $}
			\BIC{$\fifi; \Gamma \p (\nu a) P \C K$}
			\UIC{$\fifi; \Up{[m,m]} \Gamma \p m : (\nu a) P \C K + [m,m]$}
		\end{prooftree}
		
		By Lemma~\ref{l:delayinginvariance}, we know that $\Up{[m,m]} T$ is reliable. So, we have:
		
		\begin{prooftree}
			\footnotesize
			\AXC{$\pi$}
			\UIC{$\fifi; \Gamma', a : T \p P \C K'$}
			\AXC{$\fifi \p \Gamma \subtype \Gamma' ; K' \subseteq K $}
			\BIC{$\fifi; \Gamma, a : T \p P \C K$}
			\UIC{$\fifi; \Up{[m,m]} (\Gamma, a : T) \p (m : P) \C K + [m,m]$}
			\AXC{$\Up{[m,m]} T$ reliable} 
			\BIC{$\fifi; \Up{[m,m]} \Gamma \p (\nu a) (m : P) \C K + [m,m]$}
		\end{prooftree}
		
		For the converse, suppose we have $\fifi; \Gamma \p (\nu a) (m :P) \C K$. Then, the typing has the shape:
		
		\begin{prooftree}
			\small 
			\AXC{$\pi$}
			\UIC{$\fifi; \Gamma', a : T' \p P \C K'$}
			\UIC{$\fifi; \Up{[m,m]} \Gamma', a : \Up{[m,m]} T' \p (m :P) \C K' + [m,m]$}
			\AXC{$(1)$}
			\BIC{$\fifi; \Gamma, a : T \p (m :P) \C K \qquad T$ reliable} 
			\UIC{$\fifi; \Gamma \p (\nu a) (m :P) \C K$}
		\end{prooftree}
		
		where $(1)$ is 
		$$\fifi \p \Gamma \subtype \Up{[m,m]} \Gamma'; T \subtype \Up{[m,m]} T'; K' + [m,m] \subseteq K$$
		As $T$ is reliable, by Lemma~\ref{l:subusageprop}, we have $\Up{[m,m]} T'$ reliable. Then, by Lemma~\ref{l:delayinginvariance}, we have $T'$ reliable. So, we give the typing:
		
		\begin{prooftree}
			\footnotesize
			\AXC{$\pi$}
			\UIC{$\fifi; \Gamma', a : T' \p P \C K'$}
			\AXC{$T'$ reliable}
			\BIC{$\fifi; \Gamma' \p (\nu a) P \C K'$}
			\UIC{$\fifi; \Up{[m,m]} \Gamma' \p m : (\nu a) P \C K' + [m,m]$}
			\AXC{$\fifi \p \Gamma \subtype \Up{[m,m]} \Gamma'; K' + [m,m] \subseteq K$}
			\BIC{$\fifi; \Gamma \p m : (\nu a) P \C K$}
		\end{prooftree}
		
	\end{itemize}
	
	This concludes the interesting base case. Symmetry and transitivity are direct, and for the cases of contextual congruence, the proof is straightforward. 
\end{proof}

\subsection{Proof of Theorem~\ref{t:parallelsubjectreduction} (Subject Reduction)} 
\label{ss:subjectreduction} 

We now detail some cases for the proof of Theorem~\ref{t:parallelsubjectreduction}. 

Note that for a process $P$, the typing system is not syntax-directed because of the subtyping rule. However, by reflexivity and transitivity of subtyping, we can always assume that a proof has exactly \emph{one} subtyping rule before any syntax-directed rule. Moreover, notice that in those kinds of proof, the top-level rule of subtyping can be ignored. Indeed, we can always simulate exactly the same subtyping rule for both $P$ and $Q$
We now proceed by doing the case analysis on the rules of Figure~\ref{f:parallelreduction}. In order to simplify the proof, we will also consider that types and indexes invariant by subtyping (like the complexity in a server) are not renamed with subtyping. Note that this only add cumbersome notations but it does not change the core of the proof:

\begin{itemize}

	\item \textbf{Case} $(n : !\In{a}{\seq{\var}}.P) \PAR (m : \Out{a}{\seq{e}}.Q) \pred (n : !\In{a}{\seq{\var}}.P) \PAR (\max(m,n) : (\psub{P}{\seq{\var}}{\seq{e}} \PAR Q))$. Consider the typing $\fifi; \Gamma_0 \PAR \Delta_0, a : \ServT{\seq{i}}{K_a}{\seq{T}}{(U_0 \PAR V_0)} \p (n : !\In{a}{\seq{\var}}.P) \PAR (m :\Out{a}{\seq{e}}. Q) \C K_0 \Ilub K_0'$. The first rule is the rule for parallel composition, then the proof is split into the two following subtree:
	
	\begin{prooftree}
		\footnotesize 
		\AXC{$\pi_P$}
		\UIC{$(\varphi,\seq{i});\Phi;\Gamma_2, a : \ServT{\seq{i}}{K_a}{\seq{T}}{U_2}, \seq{v} : \seq{T} \p P \C K_a $}
		\UIC{$\fifi; \Up{J_c} !\Gamma_2, a : \ServT{\seq{i}}{K_a}{\seq{T}}{!\IU{[0,0]}{J_c}.U_2} \p !\In{a}{\seq{\var}}.P \C [0,0] \qquad (3) $}
		\UIC{$\fifi; \Gamma_1, a : \ServT{\seq{i}}{K_a}{\seq{T}}{U_1} \p !\In{a}{\seq{\var}}.P \C K_1$}
		\UIC{$\fifi; \Up{[n,n]} \Gamma_1, a : \ServT{\seq{i}}{K_a}{\seq{T}}{\Up{[n,n]} U_1} \p n : !\In{a}{\seq{\var}}.P \C K_1 + [n,n] \qquad (4)$}
		\UIC{$\fifi; \Gamma_0, a : \ServT{\seq{i}}{K_a}{\seq{T}}{U_0} \p n : !\In{a}{\seq{\var}}.P \C K_0$}
	\end{prooftree}
	
	\begin{prooftree}
		\scriptsize
		\AXC{$\pi_e$}
		\UIC{$\fifi; \Delta_2, a : \ServT{\seq{i}}{K_a}{\seq{T}}{V_2} \p \seq{e} \COL \Isub{\seq{T}}{\seq{i}}{\seq{\INN{I}}}$}
		\AXC{$\pi_Q$}
		\UIC{$\fifi;\Delta_2', a : \ServT{\seq{i}}{K_a}{\seq{T}}{V_2'} \p Q \C K_2$}
		\BIC{$\fifi; \Up{J_c'} (\Delta_2 \PAR \Delta_2'), a : \ServT{\seq{i}}{K_a}{\seq{T}}{\OU{[0,0]}{J_c'}.(V_2 \PAR V_2')} \p \Out{a}{\seq{e}}.Q \C J_c';(K_2 \Ilub \Isub{K_a}{\seq{i}}{\seq{\INN{I}}}) \qquad (1)$}
		\UIC{$\fifi; \Delta_1, a : \ServT{\seq{i}}{K_a}{\seq{T}}{V_1} \p \Out{a}{\seq{e}}.Q \C K_1'$}
		\UIC{$\fifi; \Up{[m,m]} \Delta_1, a : \ServT{\seq{i}}{K_a}{\seq{T}}{\Up{[m,m]} V_1} \p m :\Out{a}{\seq{e}}.Q \C K_1' + [m,m] \qquad (2)$}
		\UIC{$\fifi; \Delta_0, a : \ServT{\seq{i}}{K_a}{\seq{T}}{V_0} \p m :\Out{a}{\seq{e}}. Q \C K_0'$}
	\end{prooftree}
	
	where $(1)$ is
	\footnotesize 
	$$\fifi \p \Delta_1 \subtype \Up{J_c'} (\Delta_2 \PAR \Delta_2') \qquad \fifi \p V_1 \SubU \OU{[0,0]}{J_c'}.(V_2 \PAR V_2') \qquad \fifi \q J_c';(K_2 \Ilub \Isub{K_a}{\seq{i}}{\seq{\INN{I}}}) \subseteq K_1' $$
	\normalsize 
	$(2)$ is 
	$$ \fifi \p \Delta_0 \subtype \Up{[n,n]} \Delta_1 ; V_0 \SubU \Up{[m,m]} V_1 ; K_1' + [m,m] \subseteq K_0' $$
	$(3)$ is 
	$$\fifi \p \Gamma_1 \subtype \Up{J_c} !\Gamma_2 ; U_1 \SubU !\IU{[0,0]}{J_c}.U_2 ; [0,0] \subseteq K_1$$	
	$(4)$ is 
	$$\fifi \p \Gamma_0 \subtype \Up{[n,n]} \Gamma_1 ; U_0 \SubU \Up{[n,n]} U_1 ; K_1 + [n,n] \subseteq K_0$$

	First, by the index substitution lemma (Lemma~\ref{l:indexsub}), from $\pi_P$ we obtain a proof:
	
	$$ \Isub{\pi_P}{\seq{i}}{\seq{\INN{I}}}: \qquad \fifi;\Gamma_2, a : \ServT{\seq{i}}{K_a}{\seq{T}}{U_2}, \seq{v} : \Isub{\seq{T}}{\seq{i}}{\seq{\INN{I}}} \p P \C \Isub{K_a}{\seq{i}}{\seq{\INN{I}}}  $$
	
	Since the index variables $\seq{i}$ can only be free in $\seq{T}$ and $K_a$. 
	
	Then, we know that $\Gamma_0 \PAR \Delta_0$ is defined. Moreover, we have
	\small 
	$$ \fifi \p \Gamma_0 \subtype \Up{[n,n]} \Gamma_1 \qquad \fifi \p \Gamma_1 \subtype \Up{J_c} ! \Gamma_2 \qquad \fifi \p \Delta_0 \subtype \Up{[m,m]} \Delta_1 \qquad \Delta_1 \subtype \Up{J_c'} (\Delta_2 \PAR \Delta_2') $$ 
	\normalsize
	So, for the channel and server types, in those seven contexts, the shape of the type does not change (only the usage can change). Let us look at base types. For a context $\Gamma$, we write $\Gamma^{\Nat}$ the restriction of $\Gamma$ to base types. Then, we have: 
	\small   
	$$ \Gamma_0^{\Nat} = \Delta_0^{\Nat} \qquad \fifi \p \Gamma_0^{\Nat} \subtype \Gamma_1^{\Nat} \subtype \Gamma_2^{\Nat} \qquad \fifi \p \Delta_0^{\Nat} \subtype \Delta_1^{\Nat} \subtype \Delta_2^{\Nat} \qquad  \Delta_2^{\Nat} = \Delta_2'^{\Nat} $$
	\normalsize
	Similarly, we note $\Gamma^{\nu}$ the restriction of a context to its channel and server types. Thus, we have $\Gamma = \Gamma^{\nu}, \Gamma^{\Nat}$.
	
	So, from $\pi_e$ and $\Isub{\pi_P}{\seq{i}}{\seq{\INN{I}}}$ we obtain by subtyping: 
	$$ \pi'_P: \qquad \fifi;\Gamma_0^{\Nat}, \Gamma_2^{\nu}, a : \ServT{\seq{i}}{K_a}{\seq{T}}{U_2}, \seq{v} : \Isub{\seq{T}}{\seq{i}}{\seq{\INN{I}}} \p P \C \Isub{K_a}{\seq{i}}{\seq{\INN{I}}}  $$
	$$ \pi_e': \qquad \fifi;\Gamma_0^{\Nat}, \Delta_2^{\nu}, a : \ServT{\seq{i}}{K_a}{\seq{T}}{V_2} \p \seq{e} \COL \Isub{\seq{T}}{\seq{i}}{\seq{\INN{I}}} $$
	So, we use the substitution lemma (Lemma~\ref{l:subsitution}) and we obtain:	
	$$ \pi_{sub}: \qquad \fifi;\Gamma_0^{\Nat}, (\Gamma_2^{\nu} \PAR \Delta_2^{\nu}), a : \ServT{\seq{i}}{K_a}{\seq{T}}{(U_2 \PAR V_2)} \p \psub{P}{\seq{v}}{\seq{e}} \C \Isub{K_a}{\seq{i}}{\seq{\INN{I}}}  $$
	As previously, by subtyping from $\pi_Q$, we have: 
	
	$$ \pi'_Q: \qquad \fifi;\Gamma_0^{\Nat}, \Delta_2'^{\nu}, a : \ServT{\seq{i}}{K_a}{\seq{T}}{V_2'} \p Q \C K_2$$
	
	Thus, with the parallel composition rule (as parallel composition of context is defined) and subtyping we have:
	\small 
	$$ \pi_{PQ}: \fifi;\Gamma_0^{\Nat}, (\Gamma_2^{\nu} \PAR \Delta_2^{\nu} \PAR \Delta_2'^{\nu}), a : \ServT{\seq{i}}{K_a}{\seq{T}}{(U_2 \PAR V_2 \PAR V_2')} \p (\psub{P}{\seq{v}}{\seq{e}} \PAR Q) \C K_2 \Ilub \Isub{K_a}{\seq{i}}{\seq{\INN{I}}} $$
	\normalsize
	Let us denote $M = \max(m,n)$., and we define
	$$\mathcal{E} := \Gamma_0^{\Nat}, \Up{[M,M]}(\Gamma_2^{\nu} \PAR \Delta_2^{\nu} \PAR \Delta_2'^{\nu}), a : \ServT{\seq{i}}{K_a}{\seq{T}}{\Up{[M,M]}(U_2 \PAR V_2 \PAR V_2')}$$
	 Thus, we derive the proof: 
	
	\begin{prooftree}
		\footnotesize 
		\AXC{$\pi_{PQ}$}
		\UIC{$\fifi;\Gamma_0^{\Nat}, (\Gamma_2^{\nu} \PAR \Delta_2^{\nu} \PAR \Delta_2'^{\nu}), a : \ServT{\seq{i}}{K_a}{\seq{T}}{(U_2 \PAR V_2 \PAR V_2')} \p (\psub{P}{\seq{v}}{\seq{e}} \PAR Q) \C K_2 \Ilub \Isub{K_a}{\seq{i}}{\seq{\INN{I}}}$}
		\UIC{$\fifi; \mathcal{E} \p M : (\psub{P}{\seq{v}}{\seq{e}} \PAR Q) \C (K_2 \Ilub \Isub{K_a}{\seq{i}}{\seq{\INN{I}}}) + [M,M]$}
	\end{prooftree}
	
	Now, recall that by hypothesis, $U_0 \PAR V_0$ is reliable. We have:
	\scriptsize 
	$$ \fifi \p U_0 \SubU \Up{[n,n]} U_1 \qquad \fifi \p U_1 \SubU !\IU{[0,0]}{J_c}.U_2 \qquad \fifi \p V_0 \SubU \Up{[m,m]} V_1 \qquad \fifi \p V_1 \SubU \OU{[0,0]}{J_c'}(V_2 \PAR V_2') $$
	\normalsize 	
	So, by Point~1 of Lemma~\ref{l:subusageprop}, with transitivity and parallel composition of subusage, we have: 
	$$ \fifi \p U_0 \PAR V_0 \SubU (\Up{[n,n]}U_1) \PAR (\Up{[m,m]} V_1) \SubU !\IU{[n,n]}{J_c}.U_2 \PAR \OU{[m,m]}{J_c'}(V_2 \PAR V_2')$$
	By Point~3 of Lemma~\ref{l:subusageprop}, we have $!\IU{[n,n]}{J_c}.U_2 \PAR \OU{[m,m]}{J_c'}(V_2 \PAR V_2')$ reliable. So, in particular, we have: 
	$$ \fifi \p !\IU{[n,n]}{J_c}.U_2 \PAR \OU{[m,m]}{J_c'}(V_2 \PAR V_2') \ured !\IU{[n,n]}{J_c}.U_2 \PAR \Up{[M,M]} (U_2 \PAR V_2 \PAR V_2') $$
	$$ \fifi \q [n,n] \subseteq [m,m] \IPlus J_c' \qquad \fifi \q [m,m] \subseteq [n,n] \IPlus J_c  $$ 
	Thus, we deduce immediately that neither $J_c$ or $J_c'$ are $[\infty,\infty]$ and that
	$$ \fifi \q [M,M] \subseteq [m,M] \subseteq [n,n] + J_c \qquad \fifi \q [M,M] \subseteq [n,M] \subseteq [m,m] + J_c' $$
	So, we have in particular, with Lemma~\ref{l:subusagedelay} and Point~1 of Lemma~\ref{l:subusageprop} and parallel composition:
	\footnotesize
	$$ \fifi \p \Gamma_0 \PAR \Delta_0 \subtype (\Up{[n,n]} \Gamma_1) \PAR \Up{[m,m]} \Delta_1 \subtype (\Up{[n,n] + J_c} ! \Gamma_2) \PAR (\Up{[m,m] + J_c'}  (\Delta_2 \PAR \Delta_2')) \subtype \Up{[M,M]} (!\Gamma_2 \PAR \Delta_2 \PAR \Delta_2') $$
	\normalsize
	We also have 
	$$\fifi \q (K_2 \Ilub \Isub{K_a}{\seq{i}}{\seq{\INN{I}}}) + [M,M] \subseteq (J_c';(K_2 \Ilub \Isub{K_a}{\seq{i}}{\seq{\INN{I}}}) + [m,m]) \subseteq K_0' $$
	As Thus, we simplify a bit the derivation given above, and we have: 
	
	\begin{prooftree}
		\footnotesize 
		\AXC{$\pi_{PQ}$}
		\UIC{$\fifi;\Gamma_0^{\Nat}, (\Gamma_2^{\nu} \PAR \Delta_2^{\nu} \PAR \Delta_2'^{\nu}), a : \ServT{\seq{i}}{K_a}{\seq{T}}{(U_2 \PAR V_2 \PAR V_2')} \p (\psub{P}{\seq{v}}{\seq{e}} \PAR Q) \C K_2 \Ilub \Isub{K_a}{\seq{i}}{\seq{\INN{I}}}$}
		\UIC{$\fifi;\mathcal{E} \p M : (\psub{P}{\seq{v}}{\seq{e}} \PAR Q) \C (K_2 \Ilub \Isub{K_a}{\seq{i}}{\seq{\INN{I}}}) + [M,M]$}
		\UIC{$\fifi;\mathcal{E} \p M : (\psub{P}{\seq{v}}{\seq{e}} \PAR Q) \C K_0'$}
	\end{prooftree}
	
	We also have the following derivation: 
	
	\begin{prooftree}
		\footnotesize 
		\AXC{$\pi_P$}
		\UIC{$(\varphi,\seq{i});\Phi;\Gamma_2, a : \ServT{\seq{i}}{K_a}{\seq{T}}{U_2}, \seq{v} : \seq{T} \p P \C K_a $}
		\UIC{$\fifi; \Up{J_c} !\Gamma_2, a : \ServT{\seq{i}}{K_a}{\seq{T}}{!\IU{[0,0]}{J_c}.U_2} \p !\In{a}{\seq{\var}}.P \C [0,0]$}
		\UIC{$\fifi; \Up{[n,n] + J_c} !\Gamma_2, a : \ServT{\seq{i}}{K_a}{\seq{T}}{!\IU{[n,n]}{J_c}.U_2} \p n : !\In{a}{\seq{\var}}.P \C [n,n]$}
		\UIC{$\fifi; \Gamma_0^{\Nat}, \Up{[M,M]} !\Gamma_2^{\nu}, a : \ServT{\seq{i}}{K_a}{\seq{T}}{!\IU{[n,n]}{J_c}.U_2} \p n : !\In{a}{\seq{\var}}.P \C K_0$}
	\end{prooftree}
	
	So, by parallel composition of those two derivation we obtain a proof of:
	\footnotesize
	$$\fifi; \Gamma_0^{\Nat}, \Up{[M,M]} (!\Gamma_2^{\nu} \PAR \Gamma_2^{\nu} \PAR \Delta_2^{\nu} \PAR \Delta_2'^{\nu}), a : \ServT{\seq{i}}{K_a}{\seq{T}}{!\IU{[n,n]}{J_c}.U_2 \PAR (\Up{[M,M]}(U_2 \PAR V_2 \PAR V_2'))} \p $$ 
	$$ (n : !\In{a}{\seq{\var}}.P) \PAR M : (\psub{P}{\seq{v}}{\seq{e}} \PAR Q) \C K_0 \Ilub K_0' $$
	\normalsize
	By Point~2 of Lemma~\ref{l:subusageprop}, there exists $W$ such that:
	
	$$ \fifi \p U_0 \PAR V_0 \ured^* W \qquad \fifi \p W \SubU !\IU{[n,n]}{J_c}.U_2 \PAR \Up{[M,M]} (U_2 \PAR V_2 \PAR V_2') $$
	
	So, by subtyping we have a proof: 
	
	$$\fifi; \Gamma_0^{\Nat}, \Gamma_0^{\nu} \PAR \Delta_0^{\nu}, a : \ServT{\seq{i}}{K_a}{\seq{T}}{W} \p (n : !\In{a}{\seq{\var}}.P) \PAR M : (\psub{P}{\seq{v}}{\seq{e}} \PAR Q) \C K_0 \Ilub K_0' $$
	
	This concludes this case. 
	
	\item \textbf{Case} $(n : \In{a}{\seq{\var}}.P) \PAR (m : \Out{a}{\seq{e}}.Q) \pred (\max(m,n) : (\psub{P}{\seq{\var}}{\seq{e}} \PAR Q))$. Consider the typing $\fifi; \Gamma_0 \PAR \Delta_0, a : \ChT{\seq{T}}{(U_0 \PAR V_0)} \p (n : \In{a}{\seq{\var}}.P) \PAR (m :\Out{a}{\seq{e}}. Q) \C K_0 \Ilub K_0'$. The first rule is the rule for parallel composition, then the proof is split into the two following subtree:
	
	\begin{prooftree}
		\small  
		\AXC{$\pi_P$}
		\UIC{$\fifi;\Gamma_2, a : \ChT{\seq{T}}{U_2}, \seq{v} : \seq{T} \p P \C K_2 $}
		\UIC{$\fifi; \Up{J_c} \Gamma_2, a : \ChT{\seq{T}}{\IU{[0,0]}{J_c}.U_2} \p \In{a}{\seq{\var}}.P \C J_c;K_2 \qquad (3) $}
		\UIC{$\fifi; \Gamma_1, a : \ChT{\seq{T}}{U_1} \p \In{a}{\seq{\var}}.P \C K_1$}
		\UIC{$\fifi; \Up{[n,n]} \Gamma_1, a : \ChT{\seq{T}}{\Up{[n,n]} U_1} \p n : \In{a}{\seq{\var}}.P \C K_1 + [n,n] \qquad (4)$}
		\UIC{$\fifi; \Gamma_0, a : \ChT{\seq{T}}{U_0} \p n : \In{a}{\seq{\var}}.P \C K_0$}
	\end{prooftree}
	
	\begin{prooftree}
		\small 
		\AXC{$\pi_e$}
		\UIC{$\fifi; \Delta_2, a : \ChT{\seq{T}}{V_2} \p \seq{e} \COL \seq{T}$}
		\AXC{$\pi_Q$}
		\UIC{$\fifi;\Delta_2', a : \ChT{\seq{T}}{V_2'} \p Q \C K_2'$}
		\BIC{$\fifi; \Up{J_c'} (\Delta_2 \PAR \Delta_2'), a : \ChT{\seq{T}}{\OU{[0,0]}{J_c'}.(V_2 \PAR V_2')} \p \Out{a}{\seq{e}}.Q \C J_c';K_2'  \qquad (1)$}
		\UIC{$\fifi; \Delta_1, a : \ChT{\seq{T}}{V_1} \p \Out{a}{\seq{e}}.Q \C K_1'$}
		\UIC{$\fifi; \Up{[m,m]} \Delta_1, a : \ChT{\seq{T}}{\Up{[m,m]} V_1} \p m :\Out{a}{\seq{e}}.Q \C K_1' + [m,m] \qquad (2)$}
		\UIC{$\fifi; \Delta_0, a : \ChT{\seq{T}}{V_0} \p m :\Out{a}{\seq{e}}. Q \C K_0'$}
	\end{prooftree}
	
	where $(1)$ is: 
	$$\fifi \p \Delta_1 \subtype \Up{J_c'} (\Delta_2 \PAR \Delta_2') \qquad \fifi \p V_1 \SubU \OU{[0,0]}{J_c'}.(V_2 \PAR V_2') \qquad \fifi \q J_c';K_2' \subseteq K_1' $$
	$(2)$ is:
	$$ \fifi \p \Delta_0 \subtype \Up{[n,n]} \Delta_1 ; V_0 \SubU \Up{[m,m]} V_1 ; K_1' + [m,m] \subseteq K_0' $$
	$(3)$ is:
	$$ \fifi \p \Gamma_1 \subtype \Up{J_c} \Gamma_2 ; U_1 \SubU \IU{[0,0]}{J_c}.U_2 ; J_c;K_2 \subseteq K_1$$
	$(4)$ is:
	$$\fifi \p \Gamma_0 \subtype \Up{[n,n]} \Gamma_1 ; U_0 \SubU \Up{[n,n]} U_1 ; K_1 + [n,n] \subseteq K_0$$
	First, we know that $\Gamma_0 \PAR \Delta_0$ is defined. Moreover, we have
	\small 
	$$ \fifi \p \Gamma_0 \subtype \Up{[n,n]} \Gamma_1 \qquad \fifi \p \Gamma_1 \subtype \Up{J_c}  \Gamma_2 \qquad \fifi \p \Delta_0 \subtype \Up{[m,m]} \Delta_1 \qquad \Delta_1 \subtype \Up{J_c'} (\Delta_2 \PAR \Delta_2') $$ 
	\normalsize
	So, for the channel and server types, in those seven contexts, the shape of the type does not change (only the usage can change). We also have:  
	\small  
	$$ \Gamma_0^{\Nat} = \Delta_0^{\Nat} \qquad \fifi \p \Gamma_0^{\Nat} \subtype \Gamma_1^{\Nat} \subtype \Gamma_2^{\Nat} \qquad \fifi \p \Delta_0^{\Nat} \subtype \Delta_1^{\Nat} \subtype \Delta_2^{\Nat} \qquad  \Delta_2^{\Nat} = \Delta_2'^{\Nat} $$
	\normalsize 
	So, from $\pi_e$ and $\pi_P$ we obtain by subtyping: 
	$$ \fifi;\Gamma_0^{\Nat}, \Gamma_2^{\nu}, a : \ChT{\seq{T}}{U_2}, \seq{v} : \seq{T} \p P \C K_2 \qquad \fifi;\Gamma_0^{\Nat}, \Delta_2^{\nu}, a : \ChT{\seq{T}}{V_2} \p \seq{e} \COL \seq{T}$$
	So, we use the substitution lemma (Lemma~\ref{l:subsitution}) and we obtain:
	$$ \fifi;\Gamma_0^{\Nat}, (\Gamma_2^{\nu} \PAR \Delta_2^{\nu}), a : \ChT{\seq{T}}{(U_2 \PAR V_2)} \p \psub{P}{\seq{v}}{\seq{e}} \C K_2  $$
	As previously, by subtyping from $\pi_Q$, we have: 
	$$ \fifi;\Gamma_0^{\Nat}, \Delta_2'^{\nu}, a : \ChT{\seq{T}}{V_2'} \p Q \C K_2'$$
	Thus, with the parallel composition rule (as parallel composition of context is defined) and subtyping we have:
	$$ \fifi;\Gamma_0^{\Nat}, (\Gamma_2^{\nu} \PAR \Delta_2^{\nu} \PAR \Delta_2'^{\nu}), a : \ChT{\seq{T}}{(U_2 \PAR V_2 \PAR V_2')} \p (\psub{P}{\seq{v}}{\seq{e}} \PAR Q) \C K_2 \Ilub K_2' $$
	Let us denote $M = \max(m,n)$. Thus, we derive the proof: 
	
	\begin{prooftree}
		\scriptsize
		\AXC{$\fifi;\Gamma_0^{\Nat}, (\Gamma_2^{\nu} \PAR \Delta_2^{\nu} \PAR \Delta_2'^{\nu}), a : \ChT{\seq{T}}{(U_2 | V_2  V_2')} \p (\psub{P}{\seq{v}}{\seq{e}} \PAR Q) \C K_2 \Ilub K_2'$}
		\UIC{$\fifi;\Gamma_0^{\Nat}, \Up{[M,M]}(\Gamma_2^{\nu} | \Delta_2^{\nu} | \Delta_2'^{\nu}), a \COL \ChT{\seq{T}}{\Up{[M,M]}(U_2 | V_2 | V_2')} \p M : (\psub{P}{\seq{v}}{\seq{e}} \PAR Q) \C (K_2 \Ilub K_2') + [M,M]$}
	\end{prooftree}
	
	Now, recall that by hypothesis, $U_0 \PAR V_0$ is reliable. We have:
	\scriptsize 
	$$ \fifi \p U_0 \SubU \Up{[n,n]} U_1 \qquad \fifi \p U_1 \SubU \IU{[0,0]}{J_c}.U_2 \qquad \fifi \p V_0 \SubU \Up{[m,m]} V_1 \qquad \fifi \p V_1 \SubU \OU{[0,0]}{J_c'}(V_2 \PAR V_2') $$
	\normalsize
	So, by Point~1 of Lemma~\ref{l:subusageprop}, with transitivity and parallel composition of subusage, we have: 
	
	$$ \fifi \p U_0 \PAR V_0 \SubU (\Up{[n,n]}U_1) \PAR (\Up{[m,m]} V_1) \SubU \IU{[n,n]}{J_c}.U_2 \PAR \OU{[m,m]}{J_c'}(V_2 \PAR V_2')$$
	
	By Point~3 of Lemma~\ref{l:subusageprop}, we have $\IU{[n,n]}{J_c}.U_2 \PAR \OU{[m,m]}{J_c'}(V_2 \PAR V_2')$ reliable. So, in particular, we have: 
	$$ \fifi \p \IU{[n,n]}{J_c}.U_2 \PAR \OU{[m,m]}{J_c'}(V_2 \PAR V_2') \ured \Up{[M,M]} (U_2 \PAR V_2 \PAR V_2') $$
	$$ \fifi \q [n,n] \subseteq [m,m] \IPlus J_c' \qquad \fifi \q [m,m] \subseteq [n,n] \IPlus J_c  $$ 
	Thus, we deduce that
	$$ \fifi \q [M,M] \subseteq [n,n] + J_c \qquad \fifi \q [M,M] \subseteq [m,m] + J_c' $$
	So, we have in particular, with Lemma~\ref{l:subusagedelay} and Point~1 of Lemma~\ref{l:subusageprop} and parallel composition:
	\footnotesize 
	$$ \fifi \p \Gamma_0 \PAR \Delta_0 \subtype (\Up{[n,n]} \Gamma_1) \PAR \Up{[m,m]} \Delta_1 \subtype (\Up{[n,n] + J_c}  \Gamma_2) \PAR (\Up{[m,m] + J_c'}  (\Delta_2 \PAR \Delta_2')) \subtype \Up{[M,M]} (\Gamma_2 \PAR \Delta_2 \PAR \Delta_2') $$
	\normalsize 
	We also have 
	$$ \fifi \q K_2 + [M,M] \subseteq J_c;K_2 + [n,n] \subseteq K_0 \qquad \fifi \q K_2' + [M,M] \subseteq J_c';K_2' + [m,m] \subseteq K_0'  $$
	
	So, we obtain directly $\fifi \q (K_2 \Ilub K_2') + [M,M] \subseteq K_0 \Ilub K_0'$
	
	Thus, we can simplify a bit the derivation given above, and we have:

	\begin{prooftree}
		\scriptsize
		\AXC{$\fifi;\Gamma_0^{\Nat}, (\Gamma_2^{\nu} \PAR \Delta_2^{\nu} \PAR \Delta_2'^{\nu}), a : \ChT{\seq{T}}{(U_2 \PAR V_2 \PAR V_2')} \p (\psub{P}{\seq{v}}{\seq{e}} \PAR Q) \C K_2 \Ilub K_2'$}
		\UIC{$\fifi;\Gamma_0^{\Nat}, \Up{[M,M]}(\Gamma_2^{\nu} | \Delta_2^{\nu} | \Delta_2'^{\nu}), a : \ChT{\seq{T}}{\Up{[M,M]}(U_2 | V_2 | V_2')} \p M : (\psub{P}{\seq{v}}{\seq{e}} \PAR Q) \C (K_2 \Ilub K_2') + [M,M]$}
		\UIC{$\fifi;(\Gamma_0 \PAR \Delta_0), a : \ChT{\seq{T}}{\Up{[M,M]}(U_2 \PAR V_2 \PAR V_2')} \p M : (\psub{P}{\seq{v}}{\seq{e}} \PAR Q) \C K_0 \Ilub K_0'$}
	\end{prooftree}
	
	By Point~2 of Lemma~\ref{l:subusageprop}, there exists $W$ such that:
	
	$$ \fifi \p U_0 \PAR V_0 \ured^* W \qquad \fifi \p W \SubU \Up{[M,M]} (U_2 \PAR V_2 \PAR V_2') $$
	
	So, by subtyping we have a proof: 
	
	$$\fifi; \Gamma_0 \PAR \Delta_0, a : \ChT{\seq{T}}{W} \p M : (\psub{P}{\seq{v}}{\seq{e}} \PAR Q) \C K_0 \Ilub K_0' $$
	
	This concludes this case. 
		
	\item \textbf{Case} $\pifn{\suc(e)}{P}{x}{Q} \pred \psub{Q}{x}{e}$. The case for an expression equals to $0$ is similar, so we only present this one. Suppose given a derivation ${\pifn{\suc(e)}{P}{x}{Q}} \C K$. Then the proof has the shape:
	
	\begin{prooftree}
		\scriptsize 
		\AXC{$\pi_{e}$}
		\UIC{$\fifi;\Delta \p e \COL \Nat[I',J']$}
		\UIC{$\fifi;\Delta \p {\suc(e)} \COL \Nat[I'+1,J'+1]$}
		\AXC{$\fifi \p \Gamma \subtype \Delta; \Nat[I'+1,J'+1] \subtype \Nat[I,J]$}
		\BIC{$\fifi;\Gamma \p {\suc(e)} \COL \Nat[I,J]$}
		\AXC{$\pi_P$}
		\AXC{$\pi_Q$}
		\TIC{$\pifn{\suc(e)}{P}{x}{Q} \C K$}
	\end{prooftree}
	
	Where $\pi_Q$ is a proof of $\phi;(\Phi,J \ge 1);\Gamma, x : \Nat[I \minus 1][J \minus 1] \p {Q} \C K$, and $\pi_P$ is a typing derivation for $P$ that does not interest us in this case.
	
	By definition of subtyping, we have:
	$$ \fifi \q I \le I' + 1 \qquad \fifi \q J' + 1 \le J $$ 
	From this, we deduce the following constraints:
	$$ \fifi \q J \ge 1 \qquad \fifi \q I \minus 1 \le I' \qquad \fifi \q J' \le J \minus 1 $$
	Thus, with the subtyping rule and the proof $\pi_e$ we obtain:
	$$ \fifi;\Delta \p e \COL \Nat[I \minus 1,J \minus 1] $$
	Then, by Lemma~\ref{l:strengthening}, from $\pi_Q$ we obtain a proof of $\fifi;\Gamma, x : \Nat[I \minus 1][J \minus 1] \p {Q} \C K$. By
	the substitution lemma (Lemma~\ref{l:subsitution}), we obtain $\fifi;\Gamma \p \psub{Q}{x}{e} \C K$. This concludes this case. 
		
	\item \textbf{Case} $n : P \pred n : Q$ with $P \pred Q$. Suppose that $\fifi; \Up{[n,n]} \Gamma \p n : P \C K + [n,n]$. Then, the proof has the shape: 
	\begin{prooftree}
		\AXC{$\fifi; \Gamma \p P \C K$} 
		\UIC{$\fifi; \Up{[n,n]} \Gamma \p n : P \C K + [n,n]$} 
	\end{prooftree}  
	
	By Lemma~\ref{l:delayinginvariance}, if $\Up{[n,n]} \Gamma$ is reliable then $\Gamma$ is reliable. By induction hypothesis, we have a proof $\fifi; \Gamma' \p Q \C K$ with $\fifi \p \Gamma \ured^* \Gamma'$.
	
	We give the proof:
	\begin{prooftree}
		\AXC{$\fifi; \Gamma' \p Q \C K$} 
		\UIC{$\fifi; \Up{[n,n]} \Gamma' \p n : Q \C K + [n,n]$} 
	\end{prooftree} 
	
	And we have indeed $\fifi \p \Up{[n,n]} \Gamma \ured^* \Up{[n,n]} \Gamma'$ by Lemma~\ref{l:delayinginvariance}. 
	
	\item \textbf{Case} $P \pred Q$ with $P \congr P'$, $P' \pred Q'$ and $Q \congr Q'$. Suppose that $\fifi;\Gamma \p P \C K$. By Lemma~\ref{l:parallelcongrtype}, we have $\fifi;\Gamma \p P' \C K$. By induction hypothesis, we obtain $\fifi;\Gamma' \p Q' \C K$ with $\fifi \p \Gamma \ured^* \Gamma'$. Then, again by Lemma~\ref{l:parallelcongrtype}, we have $\fifi;\Gamma' \p Q \C K$. This concludes this case. 
\end{itemize}

This concludes the proof of Theorem~\ref{t:parallelsubjectreduction}.

\end{document}